\documentclass[a4paper,12pt]{amsart}
\usepackage[utf8]{inputenc}
\usepackage{amscd}
\usepackage{amsmath}

\usepackage{amsthm}

\usepackage{graphicx}

\usepackage{txfonts}
\usepackage{dsfont}
\usepackage{mathabx}

\usepackage{tikz-cd}
\usepackage[titletoc]{appendix}

\usepackage{geometry}
\usepackage{tikz}

\theoremstyle{definition}
\newtheorem{Defn}{Definition}[section]

\theoremstyle{plain}
\newtheorem{Lemma}[Defn]{Lemma}
\newtheorem{prop}[Defn]{Proposition}
\newtheorem{theorem}[Defn]{Theorem}
\newtheorem{Corollary}[Defn]{Corollary}

\theoremstyle{remark}
\newtheorem{remark}[Defn]{Remark}
\newtheorem{Example}[Defn]{Example}

\title{Kitaev's stabilizer code and chain complex theory of bicommutative Hopf algebras}
\author{MINKYU KIM}
\date{}

\address{Graduate School of Mathematical Sciences \\ University of Tokyo}
\email{kim@ms.u-tokyo.ac.jp}

\begin{document}

\begin{abstract}
In this paper, we give a generalization of Kitaev's stabilizer code based on chain complex theory of {\it bicommutative} Hopf algebras.
Due to the bicommutativity, the Kitaev's stabilizer code extends to a broader class of spaces, e.g. finite CW-complexes ; more generally short abstract complex over a commutative unital ring $R$ which is introduced in this paper.
Given a finite-dimensional bisemisimple bicommutative Hopf algebra with an $R$-action, we introduce some analogues of $\mathds{A}$-stabilizers, $\mathds{B}$-stabilizers and the local Hamiltonian, which we call by the $(+)$-stabilizers, the $(-)$-stabilizers and the elementary operator respectively.
We prove that the eigenspaces of the elementary operator give an orthogonal decomposition and the ground-state space is isomorphic to the homology Hopf algebra.
In application to topology, we propose a formulation of topological local stabilizer models in a functorial way.
It is known that the ground-state spaces of Kitaev's stabilizer code extends to Turaev-Viro TQFT.
We prove that the 0-eigenspaces of a topological local stabilizer model extends to a projective TQFT which is improved to a TQFT in typical examples.
Furthermore, we give a generalization of the duality in the literature based on the Poincar\'e-Lefschetz duality of R-oriented manifolds.
\end{abstract}

\maketitle

\tableofcontents

\section{Introduction}

In topological quantum computation, Kitaev's stabilizer code (a.k.a. {\it the toric code}) is a stabilizer code over polyhedral surfaces which induces a Hamiltonian model of anyons and gives a fault-tolerant quantum computation \cite{BalKir} \cite{BMCA} \cite{Kit}.
It is constructed from a finite-dimensional semisimple Hopf algebra $A$ over the complex field.
The total qudit space or crude Hilbert space is given by a tensor product of $A$'s running over $1$-cells of the surface.
The Hamiltonian is defined by a sum of stabilizers, so called $\mathds{A}$-stabilizers and $\mathds{B}$-stabilizers which are defined by using some structure on the Hopf algebra.
The ground-state space (the 0-eigenspace) of the Hamiltonian is known to be a topological invariant of the surface.

In this paper, we give a generalization of Kitaev's stabilizer code based on chain complex theory of {\it bicommutative} Hopf algebras over an arbitrary field $k$.
Due to the bicommutativity, the Kitaev's stabilizer code extends to a broader class of `spaces', for example finite CW-complexes ; more generally,  {\it short abstract complexes over a commutative unital ring $R$}.
The latter notion is introduced in section \ref{202007061028}.

Let $A$ be a finite-dimensional bisemisimple bicommutative Hopf algebra over $k$.
The Hopf algebra $A$ plays a role of a qudit or a qubit in the literature.
For a short abstract complex $X$ over $R$, we are given a vector space $V( X ; A)$ by a tensor product of $A$.
By using an $R$-action $\phi$ on $A$, we introduce some analogues of the two types of local stabilizers and the local Hamiltonian in Kitaev model : {\it the $(\pm)$-stabilizers} and {\it the elementary operator} on $V(X ; A)$ respectively.
We denote by $\mathds{H} ( X ; A , \phi)$ the elementary operator.

There is a natural chain complex $C_\bullet ( X ; A , \phi )$ of bicommutative Hopf algebras induced by $X$ and $(A, \phi)$.
The chain complex theory of bicommutative Hopf algebras is considered based on the fact that the category $\mathsf{Hopf}^\mathsf{bc}_k$ of bicommutative Hopf algebras over $k$ is an abelian category \cite{newman1975correspondence} \cite{takeuchi1972correspondence}.
We denote by $H(X ; A , \phi)$ the homology Hopf algebra which is an analogue of the homology group of chain complexes of abelian groups.

The main theorem of this paper is given as follows.
See Theorem \ref{202006302032} for the proof.
\begin{theorem}
\label{202007052132}
Let $A$ be a finite-dimensional bisemisimple bicommutative Hopf algebra over $k$ equipped with an $R$-action $\phi$.
For a short abstract complex $X$ over $R$, the following statements hold.
\begin{enumerate}
\item
The eigenspaces of the elementary operator $\mathds{H} ( X ; A , \phi)$ give a direct sum decomposition of $V(X ; A)$.
\item
$0 \in k$ is an eigenvalue of $\mathds{H} ( X ; A , \phi)$.
Furthermore, there exists a natural isomorphism of the 0-eigenspace and the homology Hopf algebra $H(X ; A, \phi)$.
\end{enumerate}
\end{theorem}

In the context of topological physics, the ground field $k$ is assumed to be the complex field.
If $A$ is a Hopf $\ast$-algebra (over $\mathbb{C}$), then $A$ has a Hermitian Frobenius form.
Moreover it is a Hermitian inner product if $A$ is a Hopf C$^\ast$-algebra.
The Hermitian Frobenius form induces a Hermitian form on $V( X ; A)$.
The $(\pm)$-stabilizers and the elementary operator are Hermitian.
The main theorem is improved to a theorem which respects the Hermitian form (see Theorem \ref{202012271301}).

For general $k$, the bicommutative Hopf algebra $A$ has a symmetric Frobenius form.
It induces a nondegenerate symmetric bilinear form on $V(X ; A)$.
The $(\pm)$-stabilizers and the elementary operator are symmetric with respect to the bilinear form.
There is another improvement of the main theorem which respects the bilinear form (see Theorem \ref{202007061137}).

In application to topology, we make use of our framework to propose a formulation of local stabilizer models (LSM's).
We define a LSM over $R$ by a symmetric monoidal functor from the category of pointed finite CW-complexes and embeddings to the category of short abstract complexes over $R$ and inclusions.
We require that it preserves pushouts (see Definition \ref{202012291052}).
Then the 0-eigenspace corresponding to a pointed finite CW-complex has a Hopf algebra structure by Theorem \ref{202007052132}.

We go further by formulating {\it topological} LSM's.
For a topological LSM, we prove that the 0-eigenspace equipped with the Hopf algebra structure is a homotopy invariant.
In particular, the 0-eigenspace is independent of the choice of the CW-complex structure.

We give a computation of the $(\pm)$-stabilizers and the elementary operator for some concrete topological LSM's (see subsection \ref{202101040012}).

The isomorphism of the 0-eigenspace and the homology Hopf algebra in some LSM (see Example \ref{202101021222}) gives a discrete and finite analogue of the Hodge decomposition.
The Hodge decomposition gives an isomorphism of the kernel of the Laplacian and the de Rham cohomology theory.
The Laplacian is obtained from the de Rham cochain complex of differential forms whereas the elementary operator is constructed from the cellular chain complex.

It is known that the Kitaev model, not necessarily bicommutative, is a Hamiltonian counterpart of 3-dimensional Dijkgraaf-Witten TQFT \cite{DW} \cite{FQ} and Turaev-Viro TQFT \cite{BalKir} \cite{BWB} \cite{Kir} \cite{BalKirETQFT} \cite{TV}.
For any topological LSM, we prove that the eigenspaces corresponding to closed $(n-1)$-manifolds extend to a {\it projective} $n$-TQFT.
Moreover, for some typical topological LSM's the eigenspaces extend to an $n$-TQFT.
It gives a generalization of the relationship of bicommutative Kitaev model and Turaev-Viro TQFT.
The discussion related with the construction of TQFT's heavily depends on our other papers \cite{kim2019integrals} \cite{kim2020family}.

We give a generalization of the duality with respect to a dual complex in oriented surfaces.
The duality in \cite{BalKir} \cite{BraKit} \cite{meusburger} factors through the Poincar\'e-Lefschetz duality and {\it a transposition duality of short abstract complexes} introduced in this paper.
We extend the duality to $R$-oriented manifolds with arbitrary dimension.

There is a classical isomorphism between the first cohomology group with an abelian group $G$ and the set of isomorphism classes of principal $G$-bundle.
We believe that our study is the (co)homology theory counterpart with respect to the Hopf algebra gauge theory \cite{meusburger}.

The $(\pm)$-stabilizers gives a generalization of a part of Vrana-Farkas \cite{vrana}.
They study a local Hamiltonian model defined on $q$-cells of an arbitrary CW-complex without any finite condition.
We consider an arbitrary bisemisimple bicommutative Hopf algebra whereas their study is based on a Hopf algebra induced by an abelian group.
Moreover, we compare the ground-state space with the homology Hopf algebra.

Kitaev's stabilizer code is constructed by a semisimple Hopf (C$^\ast$-)algebra over $\mathbb{C}$ which is not necessarily bicommutative.
It raises a question whether there exists a framework containing ours and the non-bicommutative settings.

This paper is organized as follows.
In section \ref{202007062155}, we give an overview of a stabilized object and an invariant object of (co)actions.
We characterize bisemisimplicity of a Hopf algebra by isomorphisms between stabilized objects and invariant objects.
In section \ref{202007062154}, we introduce the homology Hopf algebra of a short chain complex of bicommutative Hopf algebras.
Furthermore, we introduce the ordinary homology theory with coefficients in a bicommutative Hopf algebra.
In section \ref{202007061028}, we introduce the short abstract complexes and its related basic notions.
In section \ref{202007062235}, we define the $(\pm)$-stabilizers and the elementary operator on short abstract complexes.
We give an explanation about the transposition duality.
In section \ref{202101031448}, we prove the main theorem and the improvements below the main theorem.
In section \ref{202101051043}, we give a formulation of LSM's with some examples.
We reproduce the Kitaev's stabilizers from our framework.
Moreover, we prove that the eigenspaces of a topological LSM extend to a (projective) TQFT.
In section \ref{201907220234}, we prove a Poincar\'e-Lefschetz duality of the $(\pm)$-stabilizers and the elementary operator on a polyhedral subcomplex in an $R$-oriented manifold.
In appendix \ref{201907211845}, we give an overview of the Poincar\'e-Lefschetz duality of cellular chain complexes.

\section*{Acknowledgements}
The author was supported by FMSP, a JSPS Program for Leading Graduate Schools in the University of Tokyo, and JPSJ Grant-in-Aid for Scientific Research on Innovative Areas Grant Number JP17H06461.
The author is grateful to Christine Vespa for informing him about the work of Antoine Touz{\'e} \cite{touze}.
The author appreciates the suggestion by the referees.
Especially, the discussion related with Hopf $\ast$-algebras is motivated by one of the referees.

\section*{Terminologies and notations}

{\it Hopf algebras}
\hspace{0.5cm}
For a Hopf algebra $A$, let $\eta_A, \epsilon_A , \Delta_A, \nabla_A, S_A$ be the unit, the counit, the comultiplication, the multiplication and the antipode of $A$ respectively.
For the convenience, we often uses the Sweedler notation to represent the comultiplication : for $a \in A$, 
\begin{align}\notag
\Delta_A (a) = a^{(1)} \otimes a^{(2)} .
\end{align}
Note that it does not mean there are $a^{(1)} , a^{(2)} \in A$ whose tensor product coincides with $\Delta_A (a)$.

For a group $G$, we denote by $kG$ the induced group Hopf algebra over a field $k$.
If $G$ is finite, then we denote by $k^G$ the induced $k$-valued function Hopf algebra.

{\it CW-complexes}
\hspace{0.5cm}
A $q$-cell of a CW-pair $(L, K)$ is a $q$-cell of the CW-complex $L$ which does not lie in the subcomplex $K$.
We use the terminology {\it embedding} only to represent the induced map $i : K \hookrightarrow L$ where $K$ is a subcomplex of $L$.

{\it arrows}
\hspace{0.5cm}
We use the arrow $\hookrightarrow$ only to denote an embedding of CW-complexes.
The arrow $\rightarrowtail$ and $\twoheadrightarrow$ represent an inclusion and a restriction of short abstract complexes introduced in Definition \ref{202007080920}.

\section{Stabilized space and invariant space of (co)actions}
\label{202007062155}

In this section, we give an overview of two vector spaces associated with (co)actions of Hopf algebras : a {\it stabilized} space and an {\it invariant} space.
An isomorphism between them plays an important role in the proof of our main theorem.
We give a necessary and sufficient condition for the isomorphism to hold.

\begin{Defn}
Let $X$ be a vector space and $A$ be a Hopf algebra (over $k$).
For a left action $\alpha : A \otimes X \to X$, {\it the invariant space of $\alpha$} is defined by
\begin{align}\notag
\alpha \backslash \backslash X \stackrel{\mathrm{def.}}{=} \{ x \in X ~;~ \alpha ( a \otimes x ) = \epsilon_A (a) x, ~ a \in A \} .
\end{align}
We similarly define {\it the invariant space of a right action} $\alpha^\prime : X^\prime \otimes A^\prime \to X^\prime$ and denote by $X^\prime / / \alpha^\prime$.

For a left coaction $\beta : X \to A \otimes X$, we define {\it the invariant space of $\beta$} by a quotient space of $X$ by identifying $(id_X \otimes f) \circ \beta (x)$ with $f(\eta_A) \cdot x$ for any $x \in X$ and linear functional $f : A \to k$.
We denote it by $\beta / / X$.
Analogously, we define {\it the invariant space of a right coaction} $\beta^\prime : X^\prime \to X^\prime \otimes A^\prime$ and denote by $X^\prime \backslash \backslash \beta^\prime$.
\end{Defn}

\begin{Defn}
For a left action $\alpha : A \otimes X \to X$, we define {\it the stabilized space of $\alpha$} by
the quotient space of $X$ by identifying $\alpha ( a \otimes x ) $ with $\epsilon (a) x$ for any $x \in X$ and $a \in A$.
{\it The stabilized space of a right action} $\alpha^\prime : X^\prime \otimes A^\prime \to X^\prime$ is defined similarly and denoted by $X^\prime / \alpha^\prime$.

For a left coaction $\beta : X \to A \otimes X$, we define {\it the stabilized space of $\beta$} by
\begin{align}\notag
\beta / X \stackrel{\mathrm{def.}}{=} \{ x \in X ~;~ \beta ( x) = \eta_A \otimes x \} .
\end{align}
{\it The stabilized space of a right coaction $\beta^\prime : X^\prime \to X^\prime \otimes A^\prime$} is similarly defined by the kernel space of $(\beta^\prime - id_{X^\prime} \otimes \eta_{A^\prime})$ and denoted by $X^\prime \backslash \beta^\prime$.
\end{Defn}

\begin{Defn}
\label{202012302246}
Let $\alpha : A \otimes X \to X$ be a left action.
We define a linear homomorphism $_\alpha\gamma : \alpha \backslash \backslash X \to \alpha \backslash X$ by a composition of the inclusion and the quotient map.
For a right action $\alpha^\prime : X^\prime \otimes A^\prime \to X^\prime$, we define $\gamma_{\alpha^\prime} : X^\prime / / \alpha^\prime \to X^\prime / \alpha^\prime$ in a parallel way.
\end{Defn}

The canonical homomorphism $_\alpha \gamma$ is not an isomorphism in general.
In the following proposition, we give equivalent conditions of $A$ for $_\alpha \gamma$ to be an isomorphism.

\begin{Defn}
An element $\sigma_A \in A$ is {\it a normalized integral} if $\sigma_A \cdot a = a \cdot \sigma_A  = \epsilon_A ( a ) \cdot \eta_A$ (the integral conditions) and $\epsilon_A ( \sigma_A ) = 1$ (the condition of `normalized').
A linear functional $\sigma^A : A \to k$ is {\it a normalized cointegral} if $(\sigma^A \otimes id_A ) \circ \Delta_A = ( id_A \otimes \sigma^A ) \circ \Delta_A  = \eta_A \otimes \sigma^A$ (the cointegral conditions) and $\sigma^A ( \eta_A )= 1$ (the condition of `normalized').
\end{Defn}

\begin{prop}
\label{202012302248}
Let $A$ be a finite-dimensional Hopf algebra.
Then the following conditions are equivalent with each other.
\begin{itemize}
\item
The Hopf algebra $A$ is semisimple.
\item
There exists a normalized integral $\sigma_A$ of $A$.
\item
$_\alpha\gamma$ and $\gamma_{\alpha^\prime}$ are isomorphisms for every left action $\alpha$ and right action $\alpha^\prime$ of $A$.
\end{itemize}
\end{prop}
\begin{proof}
We give a sketch of proof.
The equivalence of the first and second statements follows from the classical reference \cite{LarSwe}.

We prove the third claim from the second one.
In fact, we could construct the inverse homomorphism $\gamma^\prime : \alpha\backslash X \to \alpha \backslash \backslash X$ by $\gamma^\prime ( [x] ) = \alpha ( \sigma_A \otimes x )$ for a left action $\alpha$.
The proof for the case of a right action is given in a symmetric fashion.

We prove the second claim from the third one.
The left regular action $\alpha = \nabla_A : A \otimes A \to A$ induces an isomorphism $_\alpha \gamma : \alpha \backslash \backslash A \to \alpha \backslash A$.
There exists $\sigma \in \alpha \backslash \backslash A$ such that $_\alpha \gamma ( \sigma ) = [ \eta_A ]$.
We have $a \cdot \sigma = \epsilon_A ( a ) \sigma$ (the left integral condition) due to $\sigma \in \alpha \backslash \backslash A$.
Moreover, $( \sigma^2 - \sigma ) = ( \sigma - \eta_A ) \sigma = \left( \sum_i ( a_i x_i - \epsilon_A ( a_i ) x_i ) \right) \sigma = 0$.
Since $\sigma^2 = \epsilon_A ( \sigma ) \sigma$, we obtain $\epsilon_A ( \sigma ) = 1$.
In a parallel way, there exists $\sigma^\prime \in A$ such that $\sigma^\prime \cdot a = \epsilon_A ( a ) \cdot \sigma^\prime$ (the right integral condition) and $\epsilon_A ( \sigma^\prime ) = 1$.
We have $\sigma = \sigma^\prime$ since $\sigma = \epsilon_A ( \sigma^\prime ) \cdot \sigma = \sigma^\prime \cdot \sigma = \sigma^\prime \cdot \epsilon_A ( \sigma) = \sigma^\prime$.
Thus, $\sigma_A = \sigma = \sigma^\prime$ is a normalized integral of $A$.
\end{proof}

In a symmetric fashion, some canonical linear homomorphisms are defined for coactions.

\begin{Defn}
Let $\beta : X \to A \otimes X$ be a left coaction.
We define a linear homomorphism $^\beta\gamma : \beta /  X \to \beta / / X$ by a composition of the inclusion and the quotient map.
For a right coaction $\beta^\prime : X^\prime \to X^\prime \otimes A^\prime$, we define $\gamma^{\beta^\prime} : X^\prime \backslash \beta^\prime \to X^\prime \backslash \backslash \beta^\prime$ in a parallel way.
\end{Defn}

By combining Proposition \ref{202012302248} and its dual statement, we obtain the following result.

\begin{prop}
\label{202101041125}
Let $A$ be a finite-dimensional Hopf algebra.
Then the following conditions are equivalent with each other.
\begin{itemize}
\item
The Hopf algebra $A$ is bisemisimple, i.e. semisimple and cosemisimple.
\item
There exists a normalized integral and normalized cointegral of $A$.
\item
$_\alpha\gamma$ and $\gamma_{\alpha^\prime}$ are isomorphisms for every left action $\alpha$ and right action $\alpha^\prime$ of $A$.
Furthermore, $^\beta \gamma$ and $\gamma^{\beta^\prime}$ are isomorphisms for every left coaction $\beta$ and right coaction $\beta^\prime$ of $A$.
\end{itemize}
\end{prop}

\begin{Example}
\label{202101052222}
Let $G$ be a finite abelian group.
Then the group Hopf algebra $A= kG$ is bisemisimple if and only if the order of $G$ is coprime to the characteristic of $k$.
In fact, if the latter condition holds, then the normalized integral and cointegral are given by $\sigma_{A} = |G|^{-1} \sum_{g \in G} g$ and $\sigma^A = \hat{e}$ where $\hat{e} (g ) = 1$ if $g=e$ and $\hat{e} (g ) = 0$ otherwise.
Especially, the function Hopf algebra $k^G$ is also bisemisimple since it is the dual Hopf algebra of $kG$.
\end{Example}

\begin{remark}
If $k$ is algebraically closed, e.g. $k = \mathbb{C}$, then all the bisemisimple bicommutative Hopf algebras arise as function Hopf algebras.
\end{remark}

\begin{Example}
Consider $k = \mathbb{R}$.
Let $B = \mathbb{R} G$ where $G = \mathbb{Z} /n$ for $n \geq 3$, the group Hopf algebra over $\mathbb{R}$.
Then $A = B \otimes B^\vee = \mathbb{R} G \otimes \mathbb{R}^G$ is a bisemisimple bicommutative Hopf algebra which is neither a group Hopf algebra nor a function Hopf algebra \cite{kim2020homology}.
\end{Example}

\begin{Example}
\label{202101052257}
Consider $k= \mathbb{F}_3$.
The direct sum algebra $\mathbb{F}_3 \oplus \mathbb{F}_9 \oplus \mathbb{F}_9$ has two bisemisimple bicommutative Hopf algebras $D_1 ,D_2$ on $\mathbb{F}_3$ \cite{kim2020homology}.
Note that $D_1, D_2$ are isomorphic via the switch of $\mathbb{F}_9$ component.
It is also neither a group Hopf algebra nor a function Hopf algebra.
\end{Example}


\section{Homology Hopf algebra}
\label{202007062154}

The chain complex theory of abelian groups could be generalized to an abstract setting.
An abelian category $\mathcal{A}$ is a category where there exists a zero object, all biproducts, all kernels and all cokernels in the sense of category theory, and the fundamental theorem on homomorphisms holds.
A chain complex in $\mathcal{A}$ is a sequence of morphisms $\{ \partial_q : C_q \to C_{q-1} \}_{q \in \mathbb{Z}}$ such that $\partial_{q-1} \circ \partial_q$ is the zero morphism whose existence is implied by the axioms of abelian category.
Likewise, the basic notions such as {\it chain complex}, {\it chain map}, {\it chain homotopy} and {\it homology} are defined formally the same as those of abelian groups.
In this section, we introduce a chain complex theory of bicommutative Hopf algebras based on following Proposition \ref{202101042026}.

\subsection{The category of bicommutative Hopf algebras}

\begin{prop}
\label{202101042026}
Let $k$ be a field.
The category $\mathsf{Hopf}^\mathsf{bc}_k$ of bicommutative Hopf algebras over $k$ is an abelian category.
\end{prop}
\begin{proof}
The proof follows from \cite{newman1975correspondence} or \cite{takeuchi1972correspondence}.
\end{proof}

\begin{remark}
\label{202012251341}
We give some remarks on the abelian category structure of $\mathsf{Hopf}^\mathsf{bc}_k$.
\begin{itemize}
\item
The ground field $k$ has a unique Hopf algebra structure which is bicommutative.
If we denote the Hopf algebra by $k$ with abuse of notation, then $k$ is a zero object in $\mathsf{Hopf}^\mathsf{bc}_k$.
\item
The biproduct, which is a direct sum and a direct product simultaneously, is given by the tensor product of Hopf algebras.
\item
In any abelian category, all the morphism sets have an abelian group structure which is biadditive with respect to the composition.
Let $A,B$ be bicommutative Hopf algebras.
The set of Hopf homomorphisms from $A$ to $B$ becomes an abelian group by the {\it convolution} :
for $\xi , \xi^\prime : A \to B$, the convolution $\xi \ast \xi^\prime$ is given by a composition $\nabla_B \circ (\xi \otimes \xi^\prime) \circ \Delta_A$.
Then the unit which we denote by $1$ is given by  the composition $\eta_B \circ \epsilon_A$.
\item
Let $\mathcal{A}$ be an abelian category with a biproduct $\oplus$.
For an object $A$ in $\mathcal{A}$, let $End_\mathcal{A} (A)$ be the endomorphism set having a natural ring structure.
There is an obvious 1-1 correspondence between $m \times n$ matrices with entries in $End_\mathcal{A} (A)$ and morphisms from $A^{\oplus n}$ to $A^{\oplus m}$.
A $(i,j)$-component of a morphism from $A^{\oplus n}$ to $A^{\oplus m}$ is the $(i,j)$-entry of the corresponding matrix.
This terminology is applied to $\mathcal{A} = \mathsf{Hopf}^\mathsf{bc}_k$.
\end{itemize}
\end{remark}

We introduce the following notations to avoid confusion between (co)kernel Hopf algebras of Hopf homomorphisms and (co)kernel space of linear homomorphisms.

\begin{Defn}
Let $A,B$ be bicommutative Hopf algebras.
For a Hopf homomorphism $\xi : A \to B$, we denote {\it a cokernel Hopf algebra of $\xi$} by a pair $( \mathrm{Cok_H} (\xi ), \mathrm{cok_H} (\xi ) )$.
In other words, it satisfies the following conditions :
\begin{enumerate}
\item
$\mathrm{Cok_H} (\xi )$ is a bicommutative Hopf algebra.
\item
$\mathrm{cok_H} (\xi ) : B \to \mathrm{Cok_H} (\xi )$ is a Hopf homomorphism.
\item
$\mathrm{cok_H} ( \xi ) \circ \xi = 1$.
\item
It is {\it universal} : if a Hopf homomorphism $\varphi : B \to C$ satisfies $\varphi \circ \xi = 1$, then there exists a unique Hopf homomorphism $\bar{\varphi} : \mathrm{Cok_H}(\xi) \to C$ such that $\bar{\varphi} \circ \mathrm{cok_H} ( \xi ) = \varphi$.
\end{enumerate}
Note that a cokernel Hopf algebra is unique up to an isomorphism by the universality.
In an analogous way, we denote by $( \mathrm{Ker_H} ( \xi ) , \mathrm{ker_H} ( \xi ) )$ {\it a kernel Hopf algebra of $\xi$}.
\end{Defn}

\begin{remark}
In this paper, we denote by $\mathrm{Cok}$ or $\mathrm{Ker}$ to represent the cokernel and kernel spaces of linear homomorphisms.
Note that $\mathrm{Cok_H} ( \xi ) \neq \mathrm{Cok} ( \xi )$ and $\mathrm{Ker_H} ( \xi ) \neq \mathrm{Ker} ( \xi )$.
\end{remark}

\begin{Defn}
Let $A,B$ be Hopf algebras and $\xi : A \to B$ be a Hopf homomorphism.
We define a left action $\alpha_\xi : A \otimes B \to B$ of $A$ on $B$ by a composition $\nabla_B \circ ( \xi \otimes id_B )$.
For $\alpha = \alpha_\xi$, we simply write $_\xi \gamma = _\alpha \gamma$.
In a dual way, we define a right coaction $\beta_\xi : A \to A \otimes B$ of $B$ on $A$ by a composition $(id_A \otimes \xi ) \circ \Delta_A$.
For $\beta = \beta_\xi$, we write $\gamma^\xi = \gamma^\beta$.
\end{Defn}

\begin{prop}
\label{202101041128}
The stabilized space $\alpha_\xi \backslash B$ has a bicommutative Hopf algebra structure which gives a cokernel Hopf algebra of $\xi$.
Analogously, the stabilized space $A \backslash \beta_\xi$ has a bicommutative Hopf algebra structure which gives a kernel Hopf algebra of $\xi$.
\end{prop}
\begin{proof}
The readers are referred to \cite{newman1975correspondence}  \cite{takeuchi1972correspondence}.
\end{proof}

\subsection{Chain complex of bicommutative Hopf algebras}

In this subsection, we define a (short) chain complex of bicommutative Hopf algebras and its {\it homology Hopf algebra}.
We give some examples derived from familiar ones of abelian groups or modules.

\begin{Defn}
\label{202012261818}
Consider Hopf homomorphisms between bicommutative Hopf algebras.
\begin{align}\notag
A_+ \stackrel{\partial_+}{\to} A_\circ \stackrel{\partial_-}{\to} A_- .
\end{align}
It is {\it a (short) chain complex of bicommutative Hopf algebras} if $\partial_- \circ \partial_+ = 1$ (the differential axiom).
Here, $1$ means the zero morphism $\eta_{A_-} \circ \epsilon_{A_+}$.
We define {\it the homology Hopf algebra} $H ( A_\bullet )$ by its homology in the abelian category $\mathsf{Hopf}^\mathsf{bc}_k$.
In other words, we define $H( A_\bullet ) \stackrel{\mathrm{def.}}{=} \mathrm{Cok_H} ( \bar{\partial}_+ )$ where the Hopf homomorphism $\bar{\partial}_+ : A_+ \to \mathrm{Ker_H} ( \partial_- )$ is induced by the Hopf homomorphism $\partial_+$ due to $\partial_- \circ \partial_+ = 1$.
\end{Defn}

\begin{remark}
The terminology is motivated by the chain complex of abelian groups and its {\it homology group}.
\end{remark}

\begin{Example}
\label{202012251535}
Let $G_+ \stackrel{\partial_+}{\to} G_\circ \stackrel{\partial_-}{\to} G_-$ be a chain complex of abelian groups.
Put $A_\square = k G_\square$ for $\square = + , \circ , -$, i.e. the group Hopf algebra of $G_\square$ on the field $k$.
Then it induces a chain complex of bicommutative Hopf algebras $A_+ \stackrel{(\partial_+)_\ast}{\to} A_\circ \stackrel{(\partial_-)_\ast}{\to} A_-$.
If we denote by $H ( G_\bullet )$ the homology group of $G_\bullet$, then we have a natural isomorphism $H ( A_\bullet ) \cong k H (G_\bullet )$.
In fact the group Hopf algebra functor $k(-) : \mathsf{Ab} \to \mathsf{Hopf}^\mathsf{bc}_k$ is an exact functor where $\mathsf{Ab}$ is the category of abelian groups.
\end{Example}

\begin{Example}
\label{202012251538}
Suppose that $G_\square$ are finite groups in the above example.
Put $A^\prime_\square = k^{G_\square}$ for $\square = + , \circ , -$, i.e. the $k$-valued function Hopf algebra of $G_\square$.
Then we obtain a chain complex of bicommutative Hopf algebras $A^\prime_- \stackrel{\partial^\ast_-}{\to} A^\prime_\circ \stackrel{\partial^\ast_+}{\to} A^\prime_+$.
We also have a natural isomorphism $H(A^\prime_\bullet ) \cong k^{H(G_\bullet)}$ since the function Hopf algebra $k^{(-)} : \left( \mathsf{Ab}^\mathsf{fin} \right)^\mathsf{op} \to \mathsf{Hopf}^\mathsf{bc}_k$ is an exact functor.
\end{Example}

Before we give one more example, we recall some results in \cite{touze}.
Let $R$ be a unital ring and $\mathsf{P}_R$ be the category of finitely-generated projective $R$-modules.
A commutative exponential functor $E$ with domain $( \mathsf{P}_R , \oplus , 0 )$ and codomain $( \mathsf{Vec}_k , \otimes , k)$ induces a bicommutative Hopf algebra $E(R)$ equipped with an $R$-action.
Here, an $R$-action on a bicommutative Hopf algebra $A$ is a unital ring homomorphism $\phi : R \to \mathrm{End} ( A )$ where $\mathrm{End} (A )$ is the {\it Hopf endomorphism} set whose additive structure induced by the convolution in Remark \ref{202012251341} and the multiplication is induced by the composition.
The assignment gives a functor $\mathcal{E}$ from the category $\mathrm{Exp}_c$ of commutative exponential functors to the category of $R$-modules in $\mathsf{Hopf}^\mathsf{bc}_k$.
Due to the first main theorem of \cite{touze}, $\mathcal{E}$ yields an equivalence of categories.
On the one hand, we have an evaluating functor $\mathcal{G}$ from $\mathrm{Exp}_c$ to the category of commutative exponential functors with domain $\mathsf{P}_R$ and codomain $\mathsf{Hopf}^\mathsf{bc}_k$.

\begin{Defn}
Consider a bicommutative Hopf algebra $A$ equipped with an $R$-action $\phi$.
We define a commutative exponential functor (or a symmetric monoidal fucntor) $(A, \phi)^{(-)}$ from $\mathsf{P}_R$ to $\mathsf{Hopf}^\mathsf{bc}_k$ as follows.
Firstly, we choose a preimage $E$ with respect to $\mathcal{E}$ which is uniquely determined up to an isomorphism.
We define $(A, \phi)^{(-)} \stackrel{\mathrm{def}}{=} \mathcal{G} (E)$, i.e. $( A , \phi) ^{M} = ( \mathcal{G} ( E ) ) ( M )$ for an object $M$ of $\mathsf{P}_R$.
\end{Defn}

\begin{remark}
The definition of $(A, \phi)^{(-)}$ depends on the choice of $E$ but it is well-defined up to a natural isomorphism.
\end{remark}

\begin{Example}
\label{202012251553}
Let $M_+ \stackrel{\partial_+}{\to} M_\circ \stackrel{\partial_-}{\to} M_-$ be a chain complex of finitely-generated projective $R$-modules.
It induces a chain complex of bicommutative Hopf algebras $(A , \phi)^{M_+} \stackrel{(\partial_+)_\ast}{\to} (A , \phi)^{M_\circ} \stackrel{(\partial_-)_\ast}{\to} (A, \phi)^{M_-}$.
If $R$ is a field (in particular $\mathsf{P}_R$ is an abelian category), then we have a natural isomorphism $H( (A, \phi)^{M_\bullet} ) \cong (A , \phi)^{H(M_\bullet)}$.
In fact, the functor $(A, \phi)^{(-)}$ is an exact functor since every exact sequence in $\mathsf{P}_R$ splits and $(A, \phi)^{(-)}$ is an additive functor.
\end{Example}

\subsection{$\mathsf{Hopf}^\mathsf{bc}_k$-valued homology theory}

A generalized homology theory is an assignment of graded abelian groups (or $R$-modules) to spaces satisfying the Eilenberg-Steenrod axioms except the dimension axiom.
For an abelian category $\mathcal{A}$, one could define an $\mathcal{A}$-valued (generalized) homology theory by an assignment of graded objects in $\mathcal{A}$ to spaces satisfying the Eilenberg-Steenrod axioms except the dimension axiom.
In this subsection, we give an explanation by restricting ourselves to $\mathcal{A} = \mathsf{Hopf}^\mathsf{bc}_k$.
As an application of the previous subsection, we prove that the ordinary homology with coefficients in a bicommutative Hopf algebra theory exists.

\begin{Defn}
\label{201912312136}
A {\it $\mathsf{Hopf}^\mathsf{bc}_k$-valued homology theory} $E_\bullet = \{ E_q , \partial_q \}_{q\in \mathbb{Z}}$ is given by the following data :
\begin{enumerate}
\item
For each integer $q$, the data $E_q$ consist of two assignments :
The first one assigns a bicommutative Hopf algebra $E_q ( K , K^\prime)$ to a finite CW-pair $(K, K^\prime)$.
The second one assigns a Hopf homomorphism $E_q (f) : E_q (K, K^\prime ) \to E_q(L , L^\prime)$ to a map $f: (K, K^\prime) \to ( L, L^\prime)$ where the objects $E_q (K, K^\prime ), E_q(L , L^\prime)$ are the corresponding objects by the first assignment.
\item
For each finite CW-pair $(K, K^\prime)$, the data $\partial_q$ is a natural transformation $\partial_q : E_{q+1} ( K , K^\prime) \to E_q ( K^\prime )$ in $\mathcal{A}$ where $E_q ( K^\prime)$ is a shorthand for $E_q ( K ^\prime , \emptyset )$.
\end{enumerate}
These are subject to following conditions :
\begin{enumerate}
\item
The assignments satisfies a functoriality :
\begin{align}\notag
E_q ( Id_{K,K^\prime} ) &= Id_{E_q(K , K^\prime)} ,  \\
E_q ( g \circ f) &= E_q (g) \circ E_q (f) . \notag
\end{align}
Here, $Id_{K,K^\prime}$ is the identity on $(K,K^\prime)$, and the maps $g,f$ are composable.
\item
Let $f,g$ be maps from $(K,K^\prime)$ to $(L,L^\prime)$.
A homotopy $f \simeq g$ induces
\begin{align}\notag
E_q ( f ) = E_q ( g) . 
\end{align}
\item
A triple of finite CW-complexes $(X, L , K)$ induces an isomorphism $E_q ( K , K \cap L ) \to E_q ( K \cup L , L )$.
\item
A pair $(K,K^\prime)$ induces a long exact sequence in $\mathcal{A}$ where $i,j$ are inclusions :
\begin{align}\notag
\cdots \to E_q ( K^\prime ) \stackrel{E_q (i)}{\to} E_q ( K ) \stackrel{E_q (j)}{\to} E_q ( K , K^\prime ) \stackrel{\partial_{q-1}}{\to} E_{q-1} ( K^\prime) \to \cdots . \notag
\end{align}
\end{enumerate}
The corresponding object $E_q ( \mathrm{pt} )$ to the one-point space is called {\it the $q$-th coefficient} of homology theory.
A $\mathsf{Hopf}^\mathsf{bc}_k$-valued homology theory is {\it ordinary} if any $q$-th coefficient is isomorphic to the trivial Hopf algebra $k$ for $q \neq 0$.
For a bicommutative Hopf algebra $A$, denote by $H_\bullet ( - , - ; A)$ an ordinary homology theory whose $0$-th coefficient is isomorphic to $A$.
\end{Defn}

\begin{prop}
For any bicommutative Hopf algebra $A$, there exists an ordinary $\mathsf{Hopf}^\mathsf{bc}_k$-valued homology theory whose coefficient is isomorphic to $A$.
\end{prop}
\begin{proof}
The $q$-th homology theory $H_q ( L , K ; A)$ is obtained by the homology Hopf algebra $H(A_\bullet)$ where $A_\bullet = \left( C^\mathrm{cell}_{q+1} (L , K ; A) \to C^\mathrm{cell}_{q} ( L, K ; A) \to C^\mathrm{cell}_{q-1} ( L , K ; A) \right)$ is the cellular chain complex with coefficients in $A$.
Such a construction does not depend on the choice of the CW-complex structures up to a natural isomorphism.
The details are given in \cite{kim2020homology}.
\end{proof}

\begin{Example}
Let $G$ be an abelian group and $A = kG$ be the group Hopf algebra.
By Example \ref{202012251535}, we have $H_q ( L , K ; A) \cong k H_q (L , K ; G)$, the group Hopf algebra of {\it the ordinary homology theory with coefficients in $G$}.
If $G$ is finite, then for the function Hopf algebra $A^\prime = k^G$ we have $H_q ( L , K ; A^\prime ) \cong k^{H^q (L, K ; G)}$, the function Hopf algebra of {\it the ordinary cohomology theory with coefficients in $G$}.
It follows from Example \ref{202012251538}.
\end{Example}

\begin{Example}
Let $R$ be a field.
Consider a bicommutative Hopf algebra $A$ with an $R$-action $\phi$.
Then we have $H_q ( L , K ; A) \cong (A, \phi)^{H_q ( L , K ; R)}$ by Example \ref{202012251553}.
\end{Example}

\section{Short abstract complex}
\label{202007061028}

In this section, we introduce a notion of short abstract complex which generalizes the polyhedral surfaces in Kitaev's model.
Roughly speaking, it consists of `cells' and `incidence numbers' which are essential in the construction of {\it $(\pm)$-stabilizers} in next section.
Furthermore, we define a homology Hopf algebra of short abstract complexes.
Throughout this section, let $R$ be a unital commutative ring.

\begin{Defn}
A quintuple $X = (X_+ , X_\circ , X_- , I_+ , I_- )$ is {\it a short abstract complex over $R$} if the following conditions hold :
\begin{enumerate}
\item
$X_+ , X_\circ , X_-$ are finite sets.
We call an element of $X_+, X_\circ , X_-$ by a $+$-cell, a $\circ$-cell, a $-$-cell respectively.
\item
$I_+ : X_+ \times X_\circ \to R$ and $I_- : X_\circ \times X_- \to R$ are maps.
We often use the notations $I_+ ( x_+ , x_\circ ) = [ x_+ : x_\circ ]$ and $I_- ( x_\circ , x_- ) = [ x_\circ : x_- ]$ and call them {\it incidence numbers}.
\item
$\sum_{x_\circ \in X_\circ} [ x_+ : x_\circ ] \cdot [ x_\circ : x_- ] = 0$ for any $x_+ \in X_+, x_- \in X_-$.
\end{enumerate}
\end{Defn}

\begin{remark}
The last condition equivalently means that the following homomorphisms $\partial_- , \partial_+$ form a chain complex.
In fact, a short abstract complex over $R$ is nothing but a freely and finitely generated short chain complex over $R$ with a distinguished basis.
\end{remark}

\begin{Defn}
For a short abstract complex $X$ over $R$, we define a short chain complex of $R$-modules $C_\bullet ( X )$ as follows.
It consists of three $R$-modules $C_+ ( X ), C_\circ ( X ), C_- ( X )$ where $C_\square ( X ) \stackrel{\mathrm{def.}}{=} \bigoplus_{x_\square \in X_\square} R$ for $\square = + , \circ , -$.
We define the boundary homomorphism $\partial_+ : C_+ ( X ) \to C_\circ ( X )$ by a homomorphism whose $( x_\circ , x_+ )$-component is $[ x_+ : x_\circ ]$.
Analogously, we define the boundary homomorphism $\partial_- : C_\circ ( X ) \to C_- ( X )$.
These data form a chain complex, i.e. $\partial_- \circ \partial_+ = 0$ since $X$ is an abstract complex over $R$.
Consider another ring $R^\prime$.
For a $(R^\prime , R)$-bimodule $M$, we define the $R^\prime$-module $H( X ; M )$ by the homology of the chain complex $M \otimes_R C( X )$.
We simply write $H( X ) = H(X ; R)$.
\end{Defn}

\begin{Defn}
\label{202101052236}
Let $X,Y$ be short abstract complexes over $R$.
We define {\it the product} $Y = X \times X^\prime$ by $Y_\square = X_\square \amalg X^\prime_\square$ for $\square = + , \circ , -$.
The incidence numbers are $[ x_\square : x^\prime_\triangle] = 0$ if $x_\square \in X_\square$, $x^\prime_\triangle \in X^\prime_\triangle$, and the same with those of $X$ and $X^\prime$ otherwise.
In other words, the incidence numbers of the product is determined by the chain isomorphism preserving the basis $C_\bullet (X) \oplus C_\bullet (X^\prime) \cong C_\bullet (X \times X^\prime)$.
\end{Defn}

\begin{Defn}
\label{202012251514}
Let $X$ be a short abstract complex over $R$.
We define {\it the transposition} as a short abstract complex $X^T = ( X^T_+ , X^T_\circ , X^T_- , I^\prime_+ , I^\prime_- )$ given by $X^T_+ \stackrel{\mathrm{def.}}{=} X_-$, $X^T_\circ \stackrel{\mathrm{def.}}{=} X_\circ$, $X^T_- \stackrel{\mathrm{def.}}{=} X_+$, $I^\prime_+ ( x^T_- , x^T_\circ) = I_- ( x_\circ , x_-  )$ and $I^\prime_- ( x^T_\circ , x^T_+) = I_- ( x_+ , x_\circ  )$.
Here, we use a formal notation $x^T_+ \in X^T_-$ ($x^T_- \in X^T_+$, resp.) to represent the element corresponding to $x_+ \in X_+$ ($x_- \in X_-$, resp.).
We simply write $C^\bullet (X) \stackrel{\mathrm{def.}}{=} C_\bullet (X^T)$.
\end{Defn}

\begin{Defn}
\label{202007080920}
Let $X,Y$ be short abstract complexes over $R$.
{\it An inclusion $s$ from $X$ to $Y$} is given by an injective map $s : (X_+ , X_\circ , X_-) \to (Y_+ , Y_\circ , Y_-)$ between triples of sets such that $[s(x_+) : s(x_\circ)] = [x_+ : x_\circ]$ and $[s(x_\circ) : s(x_-)] = [x_\circ : x_-]$ for any $x_+ \in X_+, x_\circ \in X_\circ , x_- \in X_-$.
Moreover, it satisfies the following {\it closed} conditions.
\begin{itemize}
\item
$[ s(x_\circ ) : y_- ] \neq 0$ then $y_- = s (x_-)$ for some $x_- \in X_-$.
\item
$[ s(x_+ ) : y_\circ ] \neq 0$ then $y_\circ = s (x_\circ)$ for some $x_\circ \in X_\circ$.
\end{itemize}
We denote by $s : X \rightarrowtail Y$.
If the injective map $s$ is a bijection, then $s : X \rightarrowtail Y$ is called {\it an isomorphism}.
For two inclusions $s_0 : X \rightarrowtail Y$ and $s_1 : Y \rightarrowtail Z$, we define their {\it composition} $s_1 \circ s_0$ as the usual composition of the maps between triples of sets.
We denote by $\mathsf{SAC}^\mathsf{inc}_R$ the category of short abstract complexes over $R$ and inclusions.
The product in Definition \ref{202101052236} gives a symmetric monoidal category structure.

{\it A restriction from $X$ to $Y$} is defined by an inclusion from $Y^T$ to $X^T$.
For an inclusion $s : Y^T \rightarrowtail X^T$, we denote by $s^T : X \twoheadrightarrow Y$ the associated restriction, and vice versa.
For two restrictions $r_0 : X \twoheadrightarrow Y$ and $r_1 : Y \twoheadrightarrow Z$, we define their {\it composition} $r_1 \diamond r_0 : X \twoheadrightarrow Z$ by $(r_1 \diamond r_0)^T = r^T_0 \circ r^T_1$.
We denote by $\mathsf{SAC}^\mathsf{res}_R$ the category of short abstract complexes over $R$ and restrictions.
The product in Definition \ref{202101052236} gives a symmetric monoidal category structure.
\end{Defn}

\begin{prop}
\label{202101072044}
An inclusion $s : X \rightarrowtail Y$ induces a chain map $C_\bullet (s) : C_\bullet (X) \to C_\bullet (Y)$ in a covariant way, i.e. $C_\bullet (s_1) \circ C_\bullet (s_0) = C_\bullet (s_1 \circ s_0)$.
Analogously, a restriction $r : X \twoheadrightarrow Y$ induces a chain map $C_\bullet (r) : C_\bullet (Y) \to C_\bullet (X)$ in a covariant way.
\end{prop}
\begin{proof}
Let $s: X \rightarrowtail Y$ be an inclusion.
Then it induces a chain map $s_\ast : C_\bullet (X) \to C_\bullet (Y)$ by $s_\ast (r (x_\circ )) = (r ( s(x_\circ) ))$.
Here, we identify $C_\circ (X)$ with the $R$-valued functions on the set $X_\circ$.
Note that if $s$ does not satisfies the closed condition then it is ill-defined as a chain map.
If $r : X \twoheadrightarrow Y$ is a restriction, then it induces an inclusion $r^T : Y^T \rightarrowtail X^T$ which induces a chain map $C^\bullet (r^T) : C^\bullet (X^T ) \to C^\bullet (Y^T)$.
By definitions, we obtain $C_\bullet (r) \stackrel{\mathrm{def.}}{=} C^\bullet (r^T) : C_\bullet (X ) \to C_\bullet (Y)$.
\end{proof}

Before we give examples of short abstract complexes, we give an overview of {\it incidence numbers of cells} in a CW-complex \cite{hatcher} \cite{massey}.
Let $q$ be a natural number.
For a $q$-cell $c_q$ and $(q-1)$-cell $c_{q-1}$ of a CW complex, the incidence number $[c_q : c_{q-1}]$ is an integer determined by the coefficient of $c_{q-1}$ in the boundary $\partial c_q$ where they are considered in the cellular chain complex.
The incidence number is zero if the cells $c_q, c_{q-1}$ do not meet but the reverse is not true.
The incidence numbers could be any integer depending on the CW-pair $(L , K)$.
See the following example.
\begin{Example}
\label{202012282309}
Let $n$ be an arbitrary natural number.
Consider a solid $n$-sided polygon $L^\prime$ whose CW-complex structure is determined by one face, $n$ edges and $n$ vertices.
We give an orientation to its boundary, and let $Q_n$ be the quotient CW-complex of $L^\prime$ obtained by identifying all the edges equipped with the orientation.
$L = Q_n$ has one face $d_2$, one edge $d_1$, and one vertex $d_0$ with $[d_1 : d_0] = 0$, $[d_2 : d_1] = n$.
For example, $Q_1$ is the 2-disk and $Q_2$ is the real projective surface.
\end{Example}
If $L$ is a polyhedral complex, then the incidence numbers are one of $0,1,(-1)$ as follows.

\begin{prop}
\label{201907171754}
Let $K$ be a polyhedral complex.
For a $q$-cell $c_q$ and $(q-1)$-cell $c_{q-1}$, the incidence number $[c_q : c_{q-1} ]$ is $0$, $1$, or $(-1)$.
It is $0$ if and only if $c_{q-1}$ is not a face of $c_q$.
If $c_{q-1}$ is a face of $c_q$, then the incidence number is determined by difference of the orientation of $c_{q-1}$ and that of the boundary of $c_q$ : $1$ if the orientations are the same and $(-1)$ otherwise.
\end{prop}
\begin{proof}
It is proved in Lemma 7.1 \cite{massey} (for regular CW-complexes).
\end{proof}

We give examples of short abstract complexes appearing in topology.

\begin{Example}
\label{202012261747}
Consider a nonnegative integer $q$.
For a finite CW-pair $(L, K)$, we denote by $X_\circ$ the set of $q$-cells of the pair $(L, K)$, i.e. $q$-cells of $K$ which does not lie in $L$.
Similarly, let $X_\pm$ be the set of $(q \pm 1)$-cells of $( L, K )$.
Let $I_+ ( x_+ , x_\circ ) \in R$ and $I_- ( x_\circ , x_-) \in R$ be the incidence number $[x_+ : x_\circ]_\mathrm{cell}$ and $[x_\circ : x_-]_\mathrm{cell}$ between cells respectively.
Then $\Xi_q ( L , K ; R) = ( X_+ , X_\circ , X_- , I_+ , I_- )$ is a short abstract complex over $R$.
It follows from a classical algebraic topology that $H ( \Xi_q ( L , K ; R) ) \cong H_q (L , K ; R)$.
\end{Example}

\begin{Example}
\label{202012261748}
Fix a pointed finite CW-complex $F$ with the basepoint $\ast$.
We define $\Xi^F_q ( L , K ; R) \stackrel{\mathrm{def.}}{=} \Xi_q ( L \times F , K \times F \cup L \times \{ \ast \} ; R)$.
Then $\Xi^F_q ( L , K ; R)$ is a short abstract complex over $R$ which is functorial with respect to the pair $(L, K)$.
Here, the functoriality means that an embedding $i : (L_0 , K_0 ) \hookrightarrow (L_1 , K_1)$ induces an inclusion $\Xi^F_q ( L_0 , K_0 ; R) \rightarrowtail \Xi^F_q ( L_1 , K_1 ; R)$ in a covariant way.
Note that if $F = S^0$, i.e. the pointed 0-sphere, then there exists an obvious isomorphism $\Xi^F_q ( L , K ; R) \cong \Xi_q ( L , K ; R)$ so that this example is a generalization of the previous one.
\end{Example}

\begin{Example}
\label{202012272320}
There is another example induced by a finite CW-pair $(L, K)$.
Let $X^\prime_\circ, X^\prime_+, X^\prime_-$ be the set of $q$-cells, $(q-1)$-cells, $(q+1)$-cells of the pair $(L , K )$.
Let $I^\prime_+ (x_+ , x_\circ ) = [x_\circ : x_+]_\mathrm{cell} , I^\prime_- (x_\circ , x_-) = [x_- : x_\circ]_\mathrm{cell} \in R$ be the incidence numbers between cells.
Then $\Xi^q ( L , K ; R) = ( X^\prime_+ , X^\prime_\circ , X^\prime_- , I^\prime_+ , I^\prime_- )$ is a short abstract complex over $R$.
Note any embedding between pairs induces a restriction (see Definition \ref{202007080920}) between the short abstract complexes in a contravariant way.
We remark that we have $\Xi_q ( L , K ; R)^T = \Xi^q ( L , K ; R)$ by definitions.
Especially, there exists a natural isomorphism $H ( \Xi^q ( L , K ; R) ) \cong H^q (L , K ; R)$ where $H^q (L , K ; R)$ denotes the ordinary cohomology theory with coefficients in $R$.
\end{Example}

\begin{Defn}
Let $A$ be a bicommutative Hopf algebra and $\phi$ be an $R$-action on $A$.
We define a (short) chain complex of $X$ with coefficients in $(A,\phi)$.
It consists of three bicommutative Hopf algebras $C_+ ( X; A , \phi )$, $C_\circ ( X ; A , \phi )$, $C_- ( X ; A , \phi)$ where $C_\square ( X ; A , \phi ) \stackrel{\mathrm{def.}}{=} \bigotimes_{x_\square \in X_\square} A$ for $\square = + , \circ , -$.
Let us recall the last part of Remark \ref{202012251341}.
We define the boundary homomorphism $\partial_+ : C_+ ( X; A , \phi ) \to C_\circ ( X; A , \phi )$ by a Hopf homomorphism whose $( x_\circ , x_+ )$-component is $\phi ( [ x_+ : x_\circ ] ) \in End (A )$.
Analogously, we define the boundary homomorphism $\partial_- : C_\circ ( X; A , \phi ) \to C_- ( X; A , \phi )$.
These data form a chain complex, i.e. $\partial_- \circ \partial_+ = 1$ since $X$ is an abstract complex over $R$.
Denote by $C_\bullet ( X ; A , \phi )$ the chain complex.
We define {\it the homology Hopf algebra of $X$ with coefficients in $(A , \phi)$} by
\begin{align}\notag
H(X ; A , \phi) \stackrel{\mathrm{def.}}{=} H ( C_\bullet ( X ; A , \phi ) ) . 
\end{align}
\end{Defn}

\begin{prop}
\label{202101072053}
The assignment $H(X ; A, \phi)$ is covariant with respect to inclusions (restrictions, resp.) of short abstract complexes.
\end{prop}
\begin{proof}
It is immediate from definitions and Proposition \ref{202101072044}.
\end{proof}

\begin{Example}
\label{202012252255}
An abelian group $G$ has a natural $\mathbb{Z}$-action, i.e. $(\phi (n) ) (g) = g^n$ for $n \in \mathbb{Z}, g\in G$.
It induces a $\mathbb{Z}$-action $\phi$ on the group Hopf algebra $A = k G$.
For a short abstract complex $X$ over $\mathbb{Z}$, the homology Hopf algebra $H(X ; A, \phi)$ is computed as the group Hopf algebra $k H(X ; G)$ by Example \ref{202012251535}.
If $G$ is a finite group, then we have a $\mathbb{Z}$-action $\phi^\prime$ on the function Hopf algebra $A^\prime = k^G$.
The homology Hopf algebra $H(X ; A^\prime , \phi^\prime )$ is computed as the function Hopf algebra $k^{H(X ; G)}$ by Example \ref{202012251538}.
\end{Example}

\begin{Example}
Suppose that $R$ is a field.
Let $A$ be a bicommutative Hopf algebra equipped with an $R$-action $\phi$.
For a short abstract complex $X$ over $R$, we have a natural isomorphism $H( X ; A , \phi) \cong ( A , \phi )^{H(X)}$ by Example \ref{202012251553}.
Moreover if we choose a basis of $H(X)$, we obtain an isomorphism $H(X) \cong R^{\oplus n(X)}$ where $n(X)$ is the dimension.
Then $( A , \phi )^{H(X)} \cong ( A , \phi )^{R^{\oplus n(X)}} \cong A^{\otimes n(X)}$ implies an isomorphism $H( X ; A , \phi) \cong A^{\otimes n(X)}$ depending on the basis of $H(X)$.
\end{Example}


\section{$(\pm)$-stabilizers induced by bicommutative Hopf algebras}
\label{202007062235}

In this section, we formulate local stabilizers and local Hamiltonian in our framework : the $(+)$-stabilizer, $(-)$-stabilizer and the elementary operator.
They are defined on short abstract complexes by a bicommutative Hopf algebra with an $R$-action.
We study the transposition duality of the local stabilizers and the elementary operator which yields their Poincar\'e-Lefschetz duality in section \ref{201907220234}.

\begin{Defn}
\label{202012282052}
Let $A$ be a bicommutative Hopf algebra.
Let $X$ be a short abstract complex over $R$.
For a short abstract complex $X$ over $R$, we define {\it the total space $V ( X ; A )$ over $X$} by the vector space $\bigotimes_{x_\circ \in X_\circ} A$.
\end{Defn}

\begin{Defn}
\label{202012282048}
Let $A$ be a finite-dimensional bisemisimple bicommutative Hopf algebra with the normalized integral $\sigma_A$ and cointegral $\sigma^A$.
Let $\phi$ be an $R$-action on $A$.
We define {\it the $(+)$-stabilizer} $\mathds{S}^+ ( X ,  x_+ ; A , \phi )$ and {\it the $(-)$-stabilizer} $ \mathds{S}^- ( X ,  x_- ; A , \phi )$ as follows.
For $a \in A$, we obtain an element of $V(X ; A )$ by iterating the comultiplication of $A$.
We make use of the Sweedler notation as follows.
Note that the Sweedler notation below is well-defined since $A$ is bicommutative.
\begin{align}\notag
\Delta^{X_\circ} ( a ) = \left( \cdots \otimes a^{(x_\circ)} \otimes \cdots \right) = \bigotimes_{x_\circ \in X_\circ} a^{(x_\circ)}
\end{align}
By using the Sweedler notation, we define an endomorphism $P = \mathds{S}^+ ( X ,  x_+ ; A , \phi )$ on $V ( X ; A )$ for $x_+ \in X_+$.
\begin{align}\notag
P ( \bigotimes_{x_\circ \in X_\circ} v_{x_\circ} ) \stackrel{\mathrm{def.}}{=}  \bigotimes_{x_\circ \in X_\circ} \left( ( \phi ( [ x_+ : x_\circ ] ) \sigma^{(x_\circ)}_A  ) \cdot v_{x_\circ} \right) .
\end{align}
It could be described by the string diagram in the left part of Figure \ref{202006261118}.
If there is no confusion, we abbreviate $\mathds{S}^+ ( x_+ ;  A, \phi)$ for $\mathds{S}^+ ( X , x_+ ;  A, \phi)$.
In a symmetric fashion, we define an endomorphism $P^\prime = \mathds{S}^- ( X , x_- ; A , \phi )$ on $V ( X ; A )$ for $x_- \in X_-$.
\begin{align}\notag
P^\prime ( \bigotimes_{x_\circ \in X_\circ} v_{x_\circ} ) \stackrel{\mathrm{def.}}{=}  \sigma^A \left( \prod_{x_\circ \in X_\circ} ( \phi ( [ x_\circ : x_- ] ) v^{(2)}_{x_\circ} ) \right) \cdot \bigotimes_{x_\circ \in X_\circ } v^{(1)}_{x_\circ} .
\end{align}
It could be described by the string diagram in the right part of Figure \ref{202006261118}.
If there is no confusion, we abbreviate $\mathds{S}^- ( x_- ;  A, \phi)$ for $\mathds{S}^- ( X , x_- ;  A, \phi)$.
\begin{figure}[ht]
  \includegraphics[width=15cm]{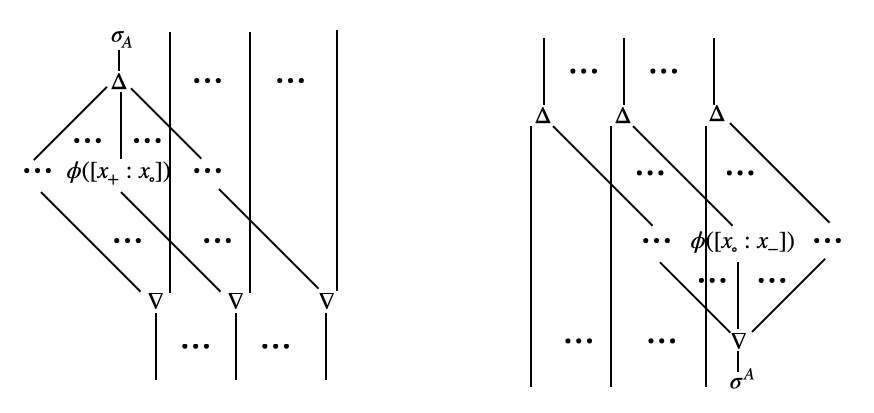}
  \caption{The stabilizers $\mathds{S}^+ ( X, x_+ ;  A , \phi)$ and $\mathds{S}^- ( X, x_- ;  A , \phi)$}
  \label{202006261118}
\end{figure}

{\it The elementary operator of $X$ induced by $(A, \phi)$} is an endomorphism $\mathds{H} ( X ; A , \phi )$ on $V ( X ; A)$ defined by the following equation.
Here, $id$ denotes the identity on the space $V( X ; A )$.
\begin{align}\notag
\mathds{H} ( X ; A , \phi ) \stackrel{\mathrm{def.}}{=} \sum_{x_+ \in X_+} ( id - \mathds{S}^+ ( x_+ ; A , \phi ) ) + \sum_{x_- \in X_-} ( id - \mathds{S}^- ( x_- ; A , \phi ) ) .
\end{align}
We define a subspace $V_0 ( X ; A , \phi)$ of $V ( X ; A )$ by the kernel space of $\mathds{H} ( X ; A , \phi )$, i.e.
\begin{align}\notag
V_0 ( X ; A , \phi) &\stackrel{\mathrm{def.}}{=} \mathrm{Ker} \left( \mathds{H} ( X ; A , \phi ) \right) , \\
&= \{ v \in V ( X ; A ) ~;~ \mathds{H} ( X ; A , \phi )v = 0 \} . \notag
\end{align}
It turns out that $V_0 ( X ; A , \phi)$ is the 0-eigenspace of the elementary operator by Theorem \ref{202006302032}.
\end{Defn}

\begin{remark}
The $(\pm)$-stabilizers and the elementary operator are symmetric or Hermitian with respect to an appropriate pairing on $V(X ; A)$.
See subsection \ref{202007062233}.
\end{remark}

\begin{remark}
\label{202012282203}
Any integer could be considered as an element of a unital ring $R$ by $n \cdot 1_R$ where $1_R$ denotes the unit of $R$.
For the convenience of the readers, we give an explanation about the $\phi$-action of $n \in \mathbb{Z}$.
\begin{align}\notag
( \phi (n) ) v = 
\begin{cases}
\epsilon_A (v) \cdot \eta_A ~&(\text{for } n =0) , \\
v ~&(\text{for } n =1), \\
S_A (v) ~&(\text{for } n =-1) .
\end{cases}
\end{align}
For arbitrary $n \neq 0$, the followings hold where we use the Sweedler notation.
\begin{align}\notag
( \phi (n) ) v = 
\begin{cases}
v^{(1)} v^{(2)} \cdots v^{(n)} ~&(\text{for } n > 0) , \\
S_A( v^{(1)} v^{(2)} \cdots v^{(-n)} ) ~&(\text{for } n < 0)
\end{cases}
\end{align}
\end{remark}

We give some examples based on a finite abelian group $G$.
For $g \in G$, we denote by $\hat{g}$ the $k$-valued function on $G$ such that $\hat{g} ( h ) = 1$ if $g = h$ and $\hat{g} ( h ) = 0$ otherwise.

\begin{Example}
\label{202012261805}
Let $A = kG$ be the group Hopf algebra of a finite abelian group.
Suppose that the characteristic of $k$ is coprime to the order of $G$.
Note that the assumption implies that $A$ is bisemisimple by Example \ref{202101052222}.
Recall the $\mathbb{Z}$-action $\phi$ on $A$ in Example \ref{202012252255}.
Let $X$ be a short abstract complex over $\mathbb{Z}$.
Then the associated stabilizers are computed as follows where $e$ denotes the unit of $G$.
\begin{align}\notag
\mathds{S}^+ ( x_+ ; A , \phi ) &: \otimes_{x_\circ} h_{x_\circ} \mapsto |G|^{-1} \sum_{g\in G} \otimes_{x_\circ} g^{[x_+ : x_\circ]} h_{x_\circ} , \notag \\
\mathds{S}^- ( x_- ; A , \phi ) &:  \otimes_{x_\circ} h_{x_\circ} \mapsto \hat{e} ( \prod_{x_\circ} h^{[x_\circ : x_-]}_{x_\circ} ) \cdot \otimes_{x_\circ} h_{x_\circ} . \notag
\end{align}
\end{Example}

\begin{Example}
Under the same hypothesis in the Example \ref{202012261805}, let $A = k^G$ the function Hopf algebra of $G$.
Consider the canonical $\mathbb{Z}$-action $\phi$.
For a $+$-cell $x_+$ of $X$, let $W(x_+)$ be the set of $\left( t_{x_\circ} \right)_{x_\circ \in X_\circ} \in G^{X_\circ}$ such that $\prod_{x_\circ} t^{[x_+ : x_\circ]}_{x_\circ} = e$.
\begin{align}\notag
\mathds{S}^+ ( x_+ ; A , \phi ) &: \otimes_{x_\circ} \hat{h}_{x_\circ} \mapsto | W(x_+) | \cdot \otimes_{x_\circ} \hat{h}_{x_\circ}  , \notag \\
\mathds{S}^- ( x_- ; A , \phi ) &:  \otimes_{x_\circ} \hat{h}_{x_\circ} \mapsto |G|^{-1} \sum_{t \in G} \sum_{g_{x_\circ} = h_{x_\circ}  t^{-[x_\circ : x_-]}} \otimes_{x_\circ} \hat{g}_{x_\circ} . \notag
\end{align}
\end{Example}

We denote by $A^\vee$ the dual Hopf algebra of $A$.
In other words, its underlying vector space is a dual of $A$ and its Hopf algebra structure is induced by the duals of that of $A$.
An $R$-action $\phi$ on $A$ induces an $R$-action $\phi^\vee$ on the dual $A^\vee$ via the duality.
Then the transposition duality of $(\pm)$-stabilizers and the elementary operator is given by the following proposition.

\begin{prop}
\label{202006302146}
Under the isomorphism of vector spaces $V(X^T ; A ) \cong V( X ; A^\vee )^\vee$, we have $\mathds{S}^+ (X^T ,  x^T_- ; A , \phi ) = \mathds{S}^- ( X , x_- ; A^\vee , \phi^\vee )^\vee$ and $\mathds{S}^- ( X^T , x^T_+ ; A , \phi ) = \mathds{S}^+ ( X , x_+ ; A^\vee , \phi^\vee )^\vee$.
In particular, we have $\mathds{H} ( X^T ; A , \phi ) = \mathds{H} ( X ; A^\vee , \phi^\vee )^\vee$.
\end{prop}
\begin{proof}
It is immediate from definitions.
\end{proof}

\section{Proof of the main theorem}
\label{202101031448}

We study the eigenspaces of the elementary operator.
In particular, we prove that the 0-eigenspace is naturally isomorphic to the homology Hopf algebra.
Note that our discussion is carried out baed on short abstract complexes, not on a special kind of Hamiltonian.
Such abstractness is useful for the construction of examples in following section.
\subsection{Eigenspaces of the elementary operator}
\label{202007062231}

In this subsection, we construct a natural linear isomorphism between the kernel space of the elementary operator and the homology Hopf algebra.
Firstly, we observe that the stabilizers are a (co)action of the normalized (co)integral through the boundary homomorphism as follows.

\begin{Defn}
\begin{enumerate}
\item
Let $\alpha : A \otimes X \to X$ be a left action of a Hopf algebra $A$ on a vector space $X$.
For linear homomorphism $a : k \to A$, we define an endomorphism $L_\alpha ( a )$ on $X$ by a composition $\alpha \circ ( a \otimes id_X )$.
\item
Let $\beta : Y \to Y \otimes B$ be a right coaction of a Hopf algebra $B$ on a vector space $Y$.
For linear homomorphism $b : B \to k$, we define an endomorphism $R^\beta ( b )$ on $Y$ by a composition $(id_Y \otimes b ) \circ \beta$.
\end{enumerate}
\end{Defn}

\begin{Lemma}
\label{202006291020}
\begin{enumerate}
\item
For $x_+ \in X_+$, we have $\mathds{S}^+ ( x_+ ; A , \phi ) = L_{\alpha_{\partial_+}} ( i_{x_+} \circ \sigma_A )$.
Here, $i_{x_+} : A \to \bigotimes_{x_+ \in X_+} A$ is the inclusion into the $x_+$-component.
\item
For $x_- \in X_-$, we have $\mathds{S}^- ( x_- ; A , \phi ) = R^{\beta_{\partial_-}} ( \sigma^A \circ p_{x_-} )$.
Here, $p_{x_-} : \bigotimes_{x_-\in X_-} A \to A$ is the projection onto the $x_-$-component.
\end{enumerate}
\end{Lemma}
\begin{proof}
The operator $\mathds{S}^+ ( x_+ ; A , \phi )$ is described by using the first string diagram in Figure \ref{202006261118}.
It is nothing but the operator $L_{\alpha_{\partial_+}} ( i_{x_+} \circ \sigma_A )$ by the definition of $\partial_+$.
The remaining part is proved analogously by using the second part in Figure \ref{202006261118}.
\end{proof}

In the computation of the kernel space of the Hamiltonian in the Kitaev model, the commutativity of stabilizers play an important role.
In our framework, we derive the commutativity from {\it the differential axiom} of chain complexes (see Definition \ref{202012261818}) :

\begin{Lemma}
\label{201907031035}
Consider a short chain complex of bicommutative Hopf algebras $C_+ \stackrel{\partial_+}{\to} C_\circ \stackrel{\partial_-}{\to} C_-$.
The induced left action $(C_+ , \alpha_{\partial_+} , C_\circ)$ and the induced right coaction $(C_\circ , \beta_{\partial_-} , C_-)$ commute with each other in the sense that the diagram below commutes :
\begin{equation}\notag
\begin{tikzcd}
C_+ \otimes C_\circ \ar[r, "id_{C_+} \otimes \beta_{\partial_-}"] \ar[d, "\alpha_{\partial_+}"] & C_+ \otimes C_\circ \otimes C_- \ar[d, "\alpha_{\partial_+} \otimes id_{C_-}"] \\
C_\circ \ar[r, "\beta_{\partial_-}"] & C_\circ \otimes C_-
\end{tikzcd}
\end{equation}
\end{Lemma}
\begin{proof}
Note that we have $\partial_- \circ \partial_+ = 1$.
Then the claim follows from the string diagrams in Figure \ref{201907021738}.
\begin{figure}[ht]
  \includegraphics[width=11cm]{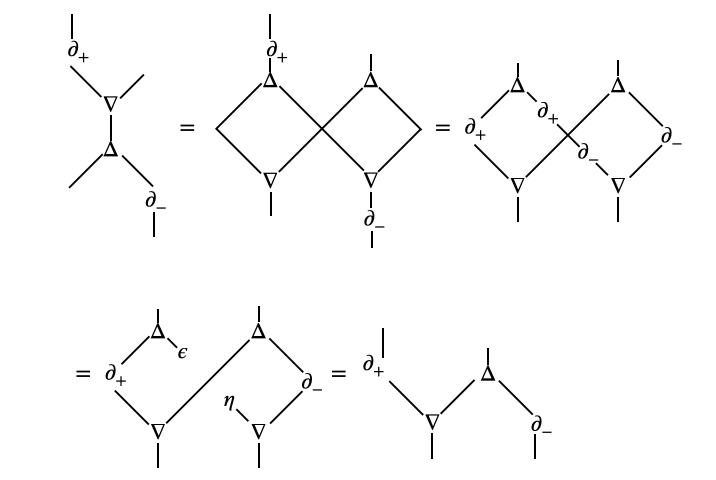}
  \caption{}
  \label{201907021738}
\end{figure}
\end{proof}

\begin{Lemma}
\label{202006291038}
All of the $(\pm)$-stabilizers $\mathds{S}^+ ( x_+ ; A , \phi )$ and $\mathds{S}^- ( x_- ; A , \phi )$ are idempotents which commute with each other.
\end{Lemma}
\begin{proof}
Note that the normalized integral $\sigma_A$ is an idempotent in $A$.
By Lemma \ref{202006291020}, the operator $\mathds{S}^+ ( x_+ ; A , \phi)$ is an idempotent.
Moreover, the operators $\mathds{S}^+ ( x_+ ; A , \phi)$ commute with each other since the Hopf algebra $A$ is commutative.
Analogously, the operators $\mathds{S}^- ( x_- ; A , \phi)$ are idempotents which commute with each other.
All that remain is to prove that $\mathds{S}^+ ( x_+ ; A , \phi)= L_{\alpha_{\partial_+}} ( i_{x_+} \circ \sigma_A )$ and $\mathds{S}^- ( x_- ; A , \phi) = R^{\beta_{\partial_-}} ( \sigma^A \circ p_{x_-} )$ commute with each other.
It follows from Lemma \ref{201907031035}.
\end{proof}

\begin{theorem}
\label{202006302032}
Let $A$ be a finite-dimensional bisemisimple bicommutative Hopf algebra with an $R$-action $\phi$.
For a short abstract complex $X$ over $R$, the following statements hold.
\begin{enumerate}
\item
The eigenspace of the elementary operator $\mathds{H} ( X ; A , \phi)$ gives a direct sum decomposition of $V(X ; A)$.
\item
$0 \in k$ is an eigenvalue of $\mathds{H} ( X ; A , \phi)$.
Furthermore, there exists a natural isomorphism of vector spaces between the 0-eigenspace $V_0 ( X ; A , \phi)$ and the homology Hopf algebra $H(X ; A, \phi)$.
\end{enumerate}
\end{theorem}
\begin{proof}
By Lemma \ref{202006291038}, the elementary operator $\mathds{H} ( X ; A , \phi )$ is a sum of commutative idempotents.
It proves the first claim.

From now on, we show that the kernel of the linear endomorphism $\mathds{H} ( X ; A , \phi )$ is naturally isomorphic to the underlying vector space of homology Hopf algebra $H(X ; A , \phi)$.
In fact, we have
\begin{align}\notag
&\mathrm{Ker} ( \mathds{H} ( X ; A , \phi ) )\\
&= 
\{ v \in V( X ; A) ~;~ \mathds{S}^+ ( x_+  ; A , \phi) v = v , ~ \mathds{S}^- ( x_-  ; A , \phi) v = v ,~ x_+ \in X_+ , ~ x_- \in X_-  \} , \notag \\
&=
\{ v \in V(X ; A ) ~;~ \prod_{x_+ \in X_+} \mathds{S}^+ ( x_+ ; A , \phi )v = v , ~ \prod_{x_- \in X_-} \mathds{S}^- ( x_- ; A , \phi )v = v \} . \notag
\end{align}
Note that Lemma \ref{202006291020} implies $\prod_{x_+ \in X_+} \mathds{S}^+ ( x_+ ; A , \phi ) = L_{\alpha_{\partial_+}} ( \sigma_{C_+ ( X ; A)} )$ and $\prod_{x_- \in X_-} \mathds{S}^- ( x_- ; A , \phi ) = R^{\beta_{\partial_-}} ( \sigma^{C_- ( X ; A)} )$.
The kernel $\mathrm{Ker} ( \mathds{H} ( X ; A , \phi ) )$ is the invariant subspace under the action of the normalized integral $\sigma_{C_+ ( X ; A)}$ and the coaction of normalized cointegral $\sigma^{C_- ( X ; A)}$.
Therefore, we obtain a linear isomorphism since the underlying vector space of $C_\circ ( X ; A , \phi )$ is $V(X ; A , \phi)$ :
\begin{align}\notag
\mathrm{Ker} ( \mathds{H} ( X ; A , \phi ) ) \cong   \alpha_{\bar{\partial}_+} \backslash \backslash \left( C_\circ ( X ; A , \phi )  \backslash \beta_{\partial_-} \right) .
\end{align}
Here, $\bar{\partial}_+ : C_+ ( X ; A, \phi) \to \mathrm{Ker_H} ( \partial_- )$ is the induced homomorphism where we identify $\mathrm{Ker_H} ( \partial_- ) \cong C_\circ ( X ; A , \phi )  \backslash \beta_{\partial_-}$ by Proposition \ref{202101041128}.
The Hopf algebra $C_+ ( X ; A , \phi)$ has a normalized integral so that $_{\bar{\partial}_+}\gamma : \alpha_{\bar{\partial}_+} \backslash \backslash \mathrm{Ker_H} ( \partial_- ) \to  \alpha_{\bar{\partial}_+} \backslash \mathrm{Ker_H} ( \partial_- )$ is an isomorphism by Proposition \ref{202101041125}.
Hence, we obtain 
\begin{align}\notag
\mathrm{Ker} ( \mathds{H} ( X ; A , \phi ) )  \cong  \alpha_{\bar{\partial}_+} \backslash \mathrm{Ker_H} ( \partial_- ) .
\end{align}
By Proposition \ref{202101041128}, we have $\alpha_{\bar{\partial}_+} \backslash \mathrm{Ker_H} ( \partial_- ) \cong \mathrm{Cok_H} ( \bar{\partial}_+ ) = H( X ; A, \phi)$.
It completes the proof.
\end{proof}

\subsection{Comparison of the pairings}
\label{202007062233}

It is known that a finite-dimensional Hopf algebra has a Frobenius form \cite{pareigis1971hopf} \cite{collins2019hopf}, i.e. a nondegenerate bilinear form $\mathrm{e} : A \times A \to k$ such that $\mathrm{e} ( ab, c) = \mathrm{e} ( a, bc)$.
It induces a bilinear form on its subspace, in particular the 0-eigenspace.
In subsection \ref{202101022034}, we prove that the 0-eigenspace {\it equipped with the induced bilinear form} does not depend on the CW-complex structure.
Furthermore, we prove that 0-eigenspace equipped with the bilinear form is isomorphic to the homology Hopf algebra equipped with the Frobenius form.

In subsection \ref{202101021119}, we focus on the case that $A$ is a Hopf $\ast$-algebra and $\phi$ is compatible with the $\ast$-structure.
In that case, we obtain a Hermitian Frobenius form, i.e. a nondegenerate Hermitian form $\mathrm{e} : A \times A \to \mathbb{C}$ such that $\mathrm{e} ( ab, c) = \mathrm{e} ( a, b^\ast c)$.
Moreover, if $A$ is a Hopf C$^\ast$-algebra then the pairing is an inner product.
We prove a theorem analogous to the above one.

\subsubsection{The symmetric bilinear forms over an arbitrary field $k$}
\label{202101022034}

For a finite-dimensional Hopf algebra $A$ with an integral $\sigma \neq 0$ and a cointegral $\sigma^\prime \neq 0$, it is known that $\sigma^\prime \circ \sigma \in k$ is invertible \cite{BKLT}.
For a finite-dimensional bisemisimple Hopf algebra $A$, we define $\mathrm{vol}^{-1} (A) = \sigma^A \circ \sigma_A$ where $\sigma_A$ is the normalized integral and $\sigma^A$ is the normalized cointegral.
Then $\mathrm{vol}^{-1} (A)$ is invertible.

\begin{Defn}
\label{202101041016}
Let $A$ be a finite-dimensional bisemisimple bicommutative Hopf algebra over $k$.
Let $\mathrm{e}_A : A \times A \to k$ be a symmetric bilinear form defined by
\begin{align}\notag
\mathrm{e}_A ( x, y) \stackrel{\mathrm{def.}}{=} \mathrm{vol}^{-1} (A)^{-1} \cdot \sigma^A ( xy ) .
\end{align}
\end{Defn}

\begin{Example}
\label{202012271516}
Consider the group Hopf algebra $A = k G$ where the characteristic of $k$ is coprime to the order of $G$.
Then $\mathrm{vol}^{-1} (A) = \hat{e} ( |G|^{-1} \sum_{g \in G} g ) = |G|^{-1}$ so that we have $\mathrm{e}_A ( g, h) = |G| \hat{e} ( gh)$ for $g , h \in G \subset A$.
\end{Example}

\begin{Example}
\label{202012271517}
Consider the function Hopf algebra $A = k^G$.
Again we obtain $\mathrm{vol}^{-1} (A) = |G|^{-1}$ so that we obtain $\mathrm{e}_A ( \hat{g} , \hat{h} ) =  \hat{e} (g^{-1} h)$.
\end{Example}

\begin{Example}
Consider $D_1$ in Example \ref{202101052257}.
We identify $\mathbb{F}_9 = \mathbb{F}_3 [\omega] / ( \omega^2 + 1)$ and put $a = (1,0,0) , b_1 = (0,1,0), b_2 = (0,\omega , 0) , c_1 = (0,0,1), c_2= (0,0,\omega) \in \mathbb{F}_3 \oplus \mathbb{F}_9 \oplus \mathbb{F}_9$.
For $x,y = a, b_1, b_2, c_1, c_2$, $e_{D_1} (x, y)$ is computed as follows.
\begin{center}
\begin{tabular}{|c|c|c|c|c|c|}
\hline
   & $a$ & $b_1$ & $b_2$ & $c_1$ & $c_2$ \\ \hline
$a$ & -1 & 0 & 0 & 0 & 0 \\
$b_1$ & 0 & 1 & 0 & 0 & 0 \\
$b_2$ & 0 & 0 & -1 & 0 & 0 \\
$c_1$ & 0 & 0 & 0 & 1 & 0 \\
$c_2$ & 0 & 0 & 0 & 0 & -1 \\
\hline
\end{tabular}
\end{center}
\end{Example}

\begin{prop}
The bilinear form $\mathrm{e}_A$ is a Frobenius form and we have $\iota ( \dim A ) = \mathrm{vol}^{-1} (A)^{-1}$ where $\iota : \mathbb{Z} \to k ; 1 \mapsto 1$.
\end{prop}
\begin{proof}
There exists a linear homomorphism $\mathrm{i}_A : k \to A \otimes A$ such that $\mathrm{e}_A ,\mathrm{i}_A$ forms a self-duality of $A$.
Such $\mathrm{i}_A$ is given by the composition $\left(  (S \otimes id_A ) \circ \Delta_A \circ \sigma_A \right)$ or, following the Sweedler notation, $\mathrm{i}_A (1) = ( \sum S(\sigma^{(1)}_A ) \otimes \sigma^{(2)}_A )$.
It is directly verified from the axioms of Hopf algebras and (co)integrals that $\mathrm{e}_A,\mathrm{i}_A$ form a self-duality where $\mathrm{e}_A : A \otimes A \to k ; a \otimes b \mapsto \mathrm{e}_A ( a , b)$ with abuse of notations.
Especially, $\mathrm{e}_A$ is nondegenerate.
It is obvious that $\mathrm{e}_A ( a b, c ) = \mathrm{e}_A ( a, bc)$ by definitions.
The {\it categorical dimension} $\mathrm{e}_A \circ \mathrm{i}_A$ should coincide with $\iota ( \dim A )$ since $\mathrm{e}_A,\mathrm{i}_A$ gives a self-duality.
We have $\mathrm{vol}^{-1} ( A) \cdot (\mathrm{e}_A \circ \mathrm{i}_A) = \sigma^A \circ \nabla_A \circ (S \otimes id_A ) \circ \Delta_A \circ \sigma_A = (\sigma^A \circ \eta_A) \cdot (\epsilon_A \circ \sigma_A) = 1$ so that $\mathrm{e}_A \circ \mathrm{i}_A = \mathrm{vol}^{-1} ( A)^{-1}$.
It completes the proof.
\end{proof}

\begin{Lemma}
\label{202007011130}
Let $A,B$ be a finite-dimensional bisemisimple bicommutative Hopf algebras.
Let $\xi : A \to B$ be a Hopf homomorphism and $\alpha_\xi : A \otimes B \to B$ be the induced left action.
Under the isomorphism $_\xi\gamma : \alpha_\xi \backslash \backslash B \to \alpha_\xi \backslash B = \mathrm{Cok_H} ( \xi )$, the restriction bilinear form of $\mathrm{e}_B$ coincides with the bilinear form $\mathrm{e}_{\mathrm{Cok_H} (\xi)}$.
\end{Lemma}
\begin{proof}
We apply some results in \cite{kim2019integrals}.
Denote by $\mu_\varphi$ the normalized generator integral along $\varphi$ for a Hopf homomorphism $\varphi$.
We first prove that $\mathrm{vol}^{-1} (B)^{-1} \cdot \sigma^B \circ \mu_{\mathrm{cok_H}(\xi)} = \mathrm{vol}^{-1} ( \mathrm{Cok_H} ( \xi ) )^{-1}  \cdot \sigma^{\mathrm{Cok_H}(\xi)}$.
Note that we have $\sigma^C = \mu_{\eta_C}$ for a Hopf algebra $C$.
Since $\eta_{\mathrm{cok_H} (\xi)} = \mathrm{cok_H} ( \xi ) \circ \eta_B$, we have $\sigma^B \circ \mu_{\mathrm{cok_H}(\xi)} = \mu_{\eta_B} \circ \mu_{\mathrm{cok_H}(\xi)} = \mathrm{vol}^{-1} ( \mathrm{Im_H} ( \xi)) \cdot \mu_{\eta_{\mathrm{Cok}_H (\xi)}} = \mathrm{vol}^{-1} ( \mathrm{Im_H} ( \xi ) ) \cdot \sigma^{\mathrm{cok_H}(\xi)}$ by the composition rule (see Theorem 12.1 \cite{kim2019integrals}).
Here, $\mathrm{Im_H} ( \xi) = \mathrm{Ker_H} ( \mathrm{cok_H} (\xi))$.
Since we have $\mathrm{vol}^{-1} ( \mathrm{Im_H} ( \xi ) ) = \mathrm{vol}^{-1} ( B ) \cdot \mathrm{vol}^{-1} ( \mathrm{Cok_H} ( \xi ) )^{-1}$ (see Corollary 12.3 \cite{kim2019integrals}), we obtain $\mathrm{vol}^{-1} ( B )^{-1} \cdot \sigma^B \circ \mu_{\mathrm{cok_H}(\xi)} = \mathrm{vol}^{-1} ( \mathrm{Cok_H} ( \xi ) )^{-1} \cdot \sigma^{\mathrm{Cok_H}(\xi)}$.

Let $x,y \in \alpha_\xi \backslash \backslash B$.
Put $z = x \cdot y$.
Then $\xi ( \sigma_A ) \cdot z =  ( \xi ( \sigma_A ) \cdot  x ) \cdot y = x \cdot y = z$ so that $z \in \alpha_\xi \backslash \backslash B$.
By Lemma 9.3 \cite{kim2019integrals}, we have $\mu_{\mathrm{cok_H}(\xi)} \circ \mathrm{cok_H} ( \xi) = L_{\alpha_\xi} ( \sigma_A )$.
Hence, we have $\mu_{\mathrm{cok_H}(\xi)} (_\xi \gamma (z)) = \mu_{\mathrm{cok_H}(\xi)} \left( \mathrm{cok_H} ( \xi ) z \right) = \left( L_{\alpha_\xi} ( \sigma_A ) \right) ( z ) = \xi ( \sigma_A ) z = z$.
By the above observation, we obtain $\mathrm{e}_B ( x , y ) = \mathrm{vol}^{-1} ( B )^{-1} \cdot \sigma^B ( z ) = \mathrm{vol}^{-1} ( B )^{-1} \cdot \sigma^B \circ \mu_{\mathrm{cok_H}(\xi)} ( _\xi \gamma (z) ) = \mathrm{vol}^{-1} ( \mathrm{Cok_H} ( \xi ) )^{-1} \cdot \sigma^{\mathrm{Cok_H}(\xi)} ( _\xi \gamma (z) ) = \mathrm{e}_{\mathrm{Cok_H}(\xi)} ( _\xi \gamma (x) , ~_\xi \gamma (y) )$.
It completes the proof.
\end{proof}

\begin{Lemma}
\label{202012271540}
Let $A$ be a finite-dimensional bisemisimple bicommutative Hopf algebra with an $R$-action $\phi$.
For a short abstract complex $X$ over $R$, the $(\pm)$-stabilizers are symmetric with respect to the induced symmetric bilinear form.
\end{Lemma}
\begin{proof}
Let $A^\prime = \bigotimes_{x_\circ} A = V(X ; A)$.
Let $P = \mathds{S}^+ ( x_+ ; A , \phi )$.
For $v, w \in V(X ;  A) (= A^\prime)$, we have
\begin{align}\notag
\mathrm{vol}^{-1} (A^\prime) \cdot e_{A^\prime} ( P x , y ) &= \sigma^{A^\prime} ( (P x) y ) , \\
&= \sigma^{A^\prime} ( ( i_{x_+} (\sigma_A ) x ) y ) ~~~(\because \mathrm{Lemma} \ref{202006291020})  , \\
&= \sigma^{A^\prime} ( x ( i_{x_+} (\sigma_A ) y) ) , \\
&= \sigma^{A^\prime} ( x ( P y) ) ~~~(\because \mathrm{Lemma} \ref{202006291020}) , \\
&= \mathrm{vol}^{-1} (A^\prime) \cdot e_{A^\prime} ( x, Py ) .
\end{align}
Analogously, one can prove that $\mathds{S}^- ( x_- ; A , \phi )$ is symmetric.
\end{proof}

\begin{theorem}
\label{202007061137}
Let $A$ be a finite-dimensional bisemisimple bicommutative Hopf algebra with an $R$-action $\phi$.
For a short abstract complex $X$ over $R$, the following statements hold.
\begin{enumerate}
\item
The eigenspaces of the elementary operator $\mathds{H} ( X ; A , \phi)$ give an orthogonal decomposition of $V(X ; A)$ with respect to the symmetric bilinear form.
\item
$0 \in k$ is an eigenvalue of $\mathds{H} ( X ; A , \phi)$.
Furthermore, there exists a natural isomorphism of vector spaces equipped with bilinear forms between the 0-eigenspace $V_0 ( X ; A , \phi)$ and the homology Hopf algebra $H(X ; A, \phi)$.
\end{enumerate}
\end{theorem}
\begin{proof}
Recall that the operator $\mathds{H} ( X ; A , \phi)$ is a sum of symmetric idempotents which commute with each other by Lemma \ref{202006291038}.
Hence the eigenspaces give an orthogonal decomposition of $V( X ; A)$.

We prove the second claim.
Recall the proof of Theorem \ref{202006302032}.
It is obvious that the bilinear form is preserved  under the isomorphism $\mathrm{Ker_H} ( \partial_- ) \cong C_\circ ( X ; A , \phi )  \backslash \beta_{\partial_-}$.
By Lemma \ref{202007011130}, the isomorphism $_{\bar{\partial}_+}\gamma : \alpha_{\bar{\partial}_+} \backslash \backslash \mathrm{Ker_H} ( \partial_- ) \to  \alpha_{\bar{\partial}_+} \backslash \mathrm{Ker_H} ( \partial_- ) \cong \mathrm{Cok} ( \bar{\partial}_+ )$ preserves the induced bilinear form.
It completes the proof.
\end{proof}

\subsubsection{The Hermitian forms}
\label{202101021119}

\begin{Defn}
Let $B$ be a Hopf algebra over $\mathbb{R}$ equipped with an anti-automorphism $\tau$ such that $\tau^2 = id_B$.
We define its {\it complexification} $B_\mathbb{C}$ by a Hopf $\ast$-algebra over $\mathbb{C}$ defined as follows.
The underlying vector space of $B_\mathbb{C}$ is $\mathbb{C} \otimes_\mathbb{R} A$ and the structure maps (i.e. the (co)multiplication, the (co)unit) are given by the coefficient extension to $\mathbb{C}$.
Here, $A$ denotes the underlying vector space of the Hopf algebra $A$.
The $\ast$-structure is given by
\begin{align}\notag
(\lambda \otimes x)^\ast \stackrel{\mathrm{def.}}{=} \bar{\lambda} \otimes \tau (x) .
\end{align}
\end{Defn}

\begin{remark}
\label{202101021132}
Any bisemisimple bicommutative Hopf $\ast$-algebra $A$ arises in this way.
In fact, any bisemisimple bicommutative Hopf algebra over $\mathbb{C}$ should be a $\mathbb{C}$-valued function Hopf algebra so that its $\ast$-structure is induced by a complexification associated with an appropriate $\tau$.
\end{remark}

\begin{remark}
\label{202101021145}
Finite-dimensional Hopf C$^\ast$-algebras are Hopf $^\ast$-algebras with a norm satisfying the C$^\ast$-identity.
By Remark \ref{202101021132}, bisemisimple and bicommutative Hopf C$^\ast$-algebras are $\mathbb{C}$-valued function Hopf algebra with the maximum norm.
\end{remark}

\begin{Defn}
\label{202101021144}
For $x,y \in B$ and $\lambda, \kappa \in \mathbb{C}$, we define a pairing on $A= B_\mathbb{C}$ as follows.
\begin{align}\notag
\mathrm{e}^\ast_A ( \lambda \otimes x , \kappa \otimes y ) \stackrel{\mathrm{def.}}{=} (\bar{\lambda} \cdot \kappa ) ~ \mathrm{e}_B ( \tau ( x ), y) \in \mathbb{C} .
\end{align}
It is obviously a Hermitian form on $A$.
In other words, it is antilinear (linear, resp.) with respect to the first (second, resp.) argument and Hermitian symmetric, i.e.
\begin{align}\notag
\overline{\mathrm{e}^\ast_A ( v , w )} = \mathrm{e}^\ast_A ( w , v ) .
\end{align}
\end{Defn}

\begin{remark}
Note that $id_\mathbb{C} \otimes \sigma^B$ is a normalized integral of $A$ which respects the $\ast$-structure.
We have $\mathrm{e}^\ast_A ( v , w ) = \mathrm{e}_A ( v^\ast , w)$ where $\mathrm{e}_A$ is the bilinear form in Definition \ref{202101041016}.
\end{remark}

\begin{remark}
Recall Remark \ref{202101021145}.
In that case, it is easy to verify that the Hermitian form is a Hermitian inner product.
Especially, the total space $V(X ; A)$ for any $X$ becomes a Hilbert space.
\end{remark}

\begin{remark}
The Hermitian form is not positive definite in general.
See the next example.
\end{remark}

\begin{Example}
Recall Example \ref{202012271516}.
For $B = \mathbb{R} G$ with $\tau = id_B$, the Hermitian form $\mathrm{e}^\ast_A$ is not positive definite.
For example, $\mathrm{e}^\ast_A ( g , g ) = 0$ if $g^2 \neq e$.
\end{Example}

\begin{Example}
Recall Example \ref{202012271517}.
For $B = \mathbb{R}^G$ with $\tau = id_B$, the Hermitian form $\mathrm{e}^\ast_A$ is positive definite so that it gives a Hermitian inner product.
\end{Example}

\begin{Defn}
For a complex vector space $V$, an antilinear automorphism $v \mapsto \overline{v}$ is called {\it a conjugate} if $\overline{\overline{v}} = v$.
Let $X$ be a short abstract complex over $R$.
Let $A$ be a complexification as above.
Then the vector space $V( X ; A)$ associated with $X$ is equipped with a conjugate by
\begin{align}
\overline{\otimes_{x_\circ} v_{x_\circ}} \stackrel{\mathrm{def.}}{=} \otimes_{x_\circ} v^\ast_{x_\circ}.
\end{align}
\end{Defn}

\begin{Defn}
Let $B$ be a finite-dimensional bisemisimple bicommutative Hopf algebra over $\mathbb{R}$ with an (anti-)automorphism $\tau$ with $\tau^2 = id_B$.
Let $\psi$ be an $R$-action compatible with $\tau$, i.e. $\psi (r) \circ \tau = \tau \circ \psi (r ) , r \in R$.
We define the complexification of $\psi$ by
\begin{align}\notag
\psi_\mathbb{C} ( r ) \stackrel{\mathrm{def.}}{=} id_\mathbb{C} \otimes \psi (r) .
\end{align}
\end{Defn}

\begin{Lemma}
\label{202012271552}
Let $X$ be a short abstract complex over $R$.
The stabilizers with respect to the complexification of $B$ are the complexification of the stabilizers with respect to $B$.
In other words, for $x_+  \in X_+$ and $x_- \in X_-$ we have
\begin{align}\notag
\mathds{S}^+ ( x_+ ; B_\mathbb{C} , \psi_\mathbb{C} ) = id_\mathbb{C} \otimes \mathds{S}^+ ( x_+ ; B , \psi ) , \\
\mathds{S}^- ( x_- ; B_\mathbb{C} , \psi_\mathbb{C} ) = id_\mathbb{C} \otimes \mathds{S}^- ( x_- ; B , \psi ) \notag
\end{align}
Moreover, the $(\pm)$-stabilizers and the elementary operator preserve the conjugate.
In particular, they are Hermitian with respect to the Hermitian form.
\end{Lemma}
\begin{proof}
We sketch the proof.
The first claim is immediate from definitions.
The $(\pm)$-stabilizers (and the elementary operator) preserve the conjugate since $\psi$ is compatible with $\tau$.
The final claim is proved by combining the first claim with Lemma \ref{202012271540}.
In fact, we have $\mathrm{e}^\ast_A ( \mathds{S}^{\pm} v , w ) = \mathrm{e}_A ( \overline{\mathds{S}^{\pm} v} , w ) = \mathrm{e}_A ( \mathds{S}^\pm \overline{v} , w ) = \mathrm{e}_A ( \overline{v} , \mathds{S}^\pm w ) = \mathrm{e}^\ast_A (v , \mathds{S}^\pm w)$.
\end{proof}

\begin{theorem}
\label{202012271301}
Let $A = B_\mathbb{C}$ and $\phi = \psi_\mathbb{C}$ be the complexification.
For a short abstract complex $X$ over $R$, the following statements hold.
\begin{enumerate}
\item
The eigenspaces of the elementary operator $\mathds{H} ( X ; A , \phi)$ give an orthogonal decomposition of $V(X ; A)$ with respect to the Hermitian form.
\item
$0 \in k$ is an eigenvalue of $\mathds{H} ( X ; A , \phi)$.
Furthermore, there exists a natural isomorphism of vector spaces equipped with Hermitian form between the 0-eigenspace $V_0 ( X ; A , \phi)$ and the homology Hopf $\ast$-algebra $H(X ; A, \phi)$.
\end{enumerate}
\end{theorem}
\begin{proof}
By Lemma \ref{202012271552}, the elementary operator is Hermitian so that the the first claim is proved.
By Theorem \ref{202007061137}, we have a $\mathbb{R}$-linear isomorphism $f : V_0 ( X ; B, \psi) \to H(X ; B , \psi)$ preserving the symmetric bilinear form.
It is obvious that the isomorphism preserves the $\mathbb{Z}/2$-action induced by $\tau$.
It induces a $\mathbb{C}$-linear isomorphism $g = (id_\mathbb{C} \otimes f) : V_0 ( X ; A, \phi) \to H(X ; A , \phi)$ preserving the conjugate where we use $V_0 (X ; A) \cong \mathbb{C} \otimes_\mathbb{R} V_0 (X ; B)$ and $H ( X ; A , \phi) \cong H( X ; B , \psi)_\mathbb{C}$.
Let $e,e^\prime$ be the symmetric bilinear form on each side.
Note that the Hermitian form is induced by applying the conjugate to the first variable.
Then we obtain $e^\prime (\overline{g(v)} , g(w) ) = e^\prime ( g(\overline{v}) , g(w)  ) = e (\overline{v} , w)$.
It completes the proof.
\end{proof}

\section{Application to topology}
\label{202101051043}
In the preceding sections, we define the $(\pm)$-stabilizers and the elementary operator associated with a short abstract complex, and study the eigenspaces on an abstract level.
In this section, we propose a formulation of {\it local stabilizer models} and {\it topological local stabilizer models}, and attempt to justify the terminologies by applying the previous results to topology.

\subsection{Local stabilizer model}
\label{201907220230}

\begin{Defn}
\label{202012291052}
{\it A local stabilizer model (LSM) over $R$} is a symmetric monoidal functor from the category of pointed finite CW-complexes and embeddings $\mathsf{CW}^\mathsf{fin,emb}_\ast$ to $\mathsf{SAC}^\mathsf{inc}_R$ (see Definition \ref{202007080920}) preserving pushouts.
We unpack the definition as follows.
It consists of an assignment of $\Xi ( K )$ to a pointed finite CW-complex $K$.
Moreover, an inclusion $\Xi ( i )  : \Xi ( K_0 ) \rightarrowtail \Xi ( K_1 )$ is given if $K_0$ is a subcomplex of $K_1$ where $i$ denotes the embedding.
They satisfy the following axioms.
\begin{itemize}
\item
The assignment is covariant, i.e. for embeddings $K_0 \stackrel{i_0}{\hookrightarrow} K_1 \stackrel{i_1}{\hookrightarrow} K_2$ we have $\Xi ( i_1 ) \circ \Xi ( i_0 ) = \Xi ( i_1 \circ i_0 )$.
Moreover, $\Xi (id_K )$ is the identity.
\item
The short abstract complex assigned to the one-point space $\ast$ (with the basepoint itself) is the trivial one, i.e. $\Xi ( \ast ) = (  \emptyset ,  \emptyset , \emptyset, \emptyset , \emptyset )$.
\item
The assignment is {\it monoidal} in the following sense.
We have a natural isomorphism $\Xi (K ) \times \Xi (K^\prime ) \cong \Xi ( K \vee K^\prime)$ where $\vee$ is the wedge sum of pointed spaces.
The natural isomorphism is compatible with the associators, unitors and symmetries.
\item
The assignment is {\it local} in the sense that the following diagram forms a pushout diagram where the arrow represents the induced injections.
\begin{equation}
\notag
\begin{tikzcd}
\Xi ( K_0 \cap K_1 ) \ar[r, rightarrowtail] \ar[d, rightarrowtail] & \Xi ( K_0 ) \ar[d, rightarrowtail] \\
\Xi ( K_1 ) \ar[r, rightarrowtail] & \Xi ( K_0 \cup K_1 )
\end{tikzcd}
\end{equation}
Equivalently, the following induced chain complex of $R$-modules is exact.
\begin{align}\notag
0 \to C_\bullet ( \Xi ( K_0 \cap K_1 ) ) \to C_\bullet ( \Xi ( K_0 ) ) \oplus C_\bullet ( \Xi ( K_1 ) ) \to C_\bullet ( \Xi ( K_0 \cup K_1 ) ) \to 0 .
\end{align}
\end{itemize}
\end{Defn}

\begin{remark}
The reasons that we formulate LSM's on pointed finite CW-complexes are as follows.
\begin{itemize}
\item
The category of pointed finite CW-complexes contains the category finite CW-complexes without a basepoint.
In fact, the assignment of $K^+ = K \amalg \{ \ast \}$ with the basepoint $\ast$ to $K$ without a basepoint induces a functor.
\item
The $(\pm)$-stabilizers and elementary operaotr could be dealt with even on finite CW-pairs, which allows us a relative theory naturally.
In fact, a finite CW-pair $(L , K)$ induces a pointed finite CW-complex $(L / K , K/ K)$.
\end{itemize}
\end{remark}

\begin{remark}
The inclusion $\Xi (i)$ induces a Hopf homomorphism $H( \Xi ( K_0 ) ; A , \phi ) \to H( \Xi ( K_1 ) ; A , \phi )$ in a functorial way.
By Theorem \ref{202006302032}, an embedding $i$ induces a map between 0-eigenspaces.
\end{remark}

We give two remarks about the locality axiom.

\begin{remark}
The locality axiom implies that the $(\pm)$-stabilizers and the elementary operator associated with $\Xi (K)$ in Definition \ref{202012282048}, in a rough sense, are determined by the {\it local} information of the CW-complex $K$.
This observation raises a question what is the minimal information to classify LSM's and topological LSM's (see Definition \ref{202012291544}).
\end{remark}

\begin{remark}
The locality axiom itself is important to prove the following Mayer-Vietoris exact sequence of homology Hopf algebras.
The Mayer-Vietoris exact sequence is important for the corresponding homology Hopf algebras (equivalently, 0-eigenspaces) to extend to a Brown functor and TQFT.
See subsection \ref{202101021342}.
\end{remark}

\begin{prop}
\label{202101021609}
For a LSM $\Xi$ over $R$, the inclusions induce the following chain complex of bicommutative Hopf algebras which is exact in $\mathsf{Hopf}^\mathsf{bc}_k$.
\begin{align}\notag
H ( \Xi ( K_0 \cap K_1 ) ; A , \phi) \to H ( \Xi ( K_0 ) ; A , \phi ) \otimes H ( \Xi ( K_1 ) ; A , \phi ) \to H ( \Xi ( K_0 \cup K_1 ) ; A , \phi )  .
\end{align}
\end{prop}
\begin{proof}
It is immediate from the locality axiom.
\end{proof}

We formulate a `topologicality' of LSM's in our framework as follows.

\begin{Defn}
\label{202012291544}
Let $A$ be a finite-dimensional bisemisimple bicommutative Hopf algebra with an $R$-action $\phi$.
A LSM $\Xi$ is {\it topological with respect to $(A, \phi)$} if the following conditions hold.
\begin{itemize}
\item
Let $i_t : K \hookrightarrow K \wedge [0,1]^+$ be the embedding $i_t (x) = [ x , t ]$ for $t=0,1$.
Then we have $H(\Xi (i_0) ; A , \phi) = H( \Xi (i_1) ; A , \phi )$.
\item
An embedding $i : K_0 \hookrightarrow K_1$ which is a homotopy equivalence induces an isomorphism $H( \Xi ( i); A , \phi) : H ( \Xi ( K_0 ) ; A , \phi ) \to H ( \Xi ( K_1 ) ; A , \phi )$.
\end{itemize}

A LSM $\Xi$ is {\it topological} if it is topological with respect to any $(A, \phi)$.
\end{Defn}

\begin{remark}
\label{202101030318}
The readers might wonder whether one of the axioms is redundant for our discussion or not.
Only the second axiom is necessary to prove the homotopy invariance of 0-eigenspaces in Theorem \ref{202101021229}.
On the other hand, the path-integrals in Definition \ref{202101030308} is a homotopy invariant of cospan if and only if both of the axioms are satisfied.
See Proposition \ref{202101022059}.
\end{remark}

\begin{Example}
Let $K$ be a pointed finite CW-complex.
We define $\Xi^\prime_q ( K ) = ( X_+ , X_\circ , X_- , I_+ , I_- )$ by $X_+ = X_-  = \emptyset$, $I_+ = I_- = \emptyset$.
We set $X_\circ$ be the set of $q$-cells of the pointed CW-complex $K$, i.e. $q$-cells of the unpointed CW-complex $K$ which is not the basepoint.
It is obvious that $\Xi^\prime_q$ gives a LSM over any ring $R$.
Note that this is an example which is {\it not topological}.
In fact, we have $V_0 \left( \mathds{H} ( \Xi^\prime_q ( K ) ; A , \phi \right) \cong H ( \Xi^\prime_q ( K ) ; A , \phi ) \cong \bigotimes A$ where the tensor product is taken over $q$-cells.
\end{Example}

\begin{Example}
\label{202101021222}
The Example \ref{202012261747} gives a LSM over $R$ by $\Xi ( K ) = \Xi_q ( K , \{ \ast \} ; R)$.
Here, $\ast$ denotes the basepoint of $K$.
Moreover it topological.
\end{Example}

\begin{Example}
\label{202007061120}
More generally, the Examples \ref{202012261748} gives a topological LSM over $R$ by $\Xi ( K ) = \Xi^F_q ( K , \{ \ast \} ; R)$.
\end{Example}

\begin{remark}
One may formulate a {\it contravariant} LSM over $R$ by modifying the functoriality and locality axioms.
We often call the LSM in Definition \ref{202012291052} by a {\it covariant} LSM to avoid confusion.
A {\it contravariant LSM over $R$} $\Xi^\prime$ assigns a short abstract complex $\Xi^\prime (K)$ to $K$ and a restriction $\Xi^\prime ( i ) : \Xi^\prime ( K_1) \twoheadrightarrow \Xi^\prime (K_0)$ (see Definition \ref{202007080920}) to an embedding $i$ such that the transpositions $\Xi^\prime (K)^T$ and $\Xi^\prime (i)^T$ give a covariant LSM over $R$.
Especially, $\Xi^\prime ( i_0) \diamond \Xi^\prime (i_1) = \Xi^\prime (i_1 \circ i_0)$.
\end{remark}

\begin{Example}
\label{202101062347}
The Example \ref{202012272320} gives a contravariant LSM over $R$.
\end{Example}

\begin{theorem}
\label{202101021229}
Let $\Xi$ be a (not necessarily topological) LSM over $R$.
\begin{itemize}
\item
The eigenspaces of the elementary operator $\mathds{H} ( \Xi ( K ) ; A , \phi )$ give an orthogonal decomposition of the total space with respect to the symmetric bilinear form.
\item
The 0-eigenspace $V_0 ( \Xi ( K ) ; A , \phi )$ has a Hopf algebra structure which is finite-dimensional, bisemisimple and bicommutative.
The bilinear form on$V_0 ( \Xi ( K ) ; A , \phi )$ coincides with the one induced by the Hopf algebra structure.
\item
If $\Xi$ is topological , then the 0-eigenspace $V_0 ( \Xi ( K ) ; A , \phi )$ equipped with the Hopf algebra structure is a homotopy invariant of $K$.
To be precise, for a pointed homotopy equivalence $f : K_0 \to K_1$ there exists a Hopf isomorphism $f_\ast : V_0 \left( \Xi ( K_0 ) ; A , \phi  \right) \to V_0 \left( \Xi ( K_1 ) ; A , \phi \right)$.
\end{itemize}
\end{theorem}
\begin{proof}
The first and second claims are immediate from Theorem \ref{202007061137}.

Suppose that $\Xi$ is topological.
It suffices to prove that $H ( \Xi ( K ) ; A , \phi )$ is a homotopy invariant.
Suppose that $f : K_0 \to K_1$ is a homotopy equivalence which is not necessarily an embedding.
Its mapping cylinder gives a pointed finite CW-complex $M(f)$ with embeddings $K_0 \stackrel{i}{\hookrightarrow} M(f) \stackrel{j}{\hookleftarrow} K_1$ which are homotopy equivalences.
Since $\Xi$ is topological and $j$ is a pointed homotopy equivalence, the induced homomorphisms $\Xi ( i ), \Xi (j)$ are isomorphisms.
Hence, we obtain an isomorphism $f_\ast = \Xi (j)^{-1} \circ \Xi (i) : H ( \Xi (K_0  ) ; A , \phi ) \to H ( \Xi (K_1 ) ; A , \phi )$.
\end{proof}

\begin{remark}
\label{202101022117}
By applying Theorem \ref{202012271301}, we obtain a theorem more familiar with the context of physics.
To be precise, for $( A, \phi) = (B_\mathbb{C} , \psi_\mathbb{C} )$ in subsubsection \ref{202101021119}, the 0-eigenspace has a Hopf $\ast$-algebra structure whose Hermitian forms coincide with each other.
If $\Xi$ is topological, then the Hopf $\ast$-algebra structure is a homotopy invariant.
The case of Hopf $C^\ast$-algebra could be dealt with as a special case (see Remark \ref{202101021145}).
\end{remark}

\begin{remark}
The Hopf isomorphism $f_\ast$ is covariant in the sense that $(g \circ f )_\ast = g_\ast \circ f_\ast$.
It follows from Lemma \ref{202012291530}, i.e. $Y(J(g)) \circ Y(J(f)) = Y(J(g \circ f))$ where $Y(J(f)) = f_\ast$ by definitions.
\end{remark}

\subsection{Computations of $(\pm)$-stabilizers}
\label{202101040012}

In this subsection, we describe the $(\pm)$-stabilizers associated with some LSM.
Fix a finite-dimensional bisemisimple bicommutative Hopf algebra $A$ with an $R$-action $\phi$.
Let $K$ be a pointed finite CW-complex.
Denote by $N_q$ the set of $q$-cells of $K$ which is not the basepoint.

Recall the LSM in Example \ref{202101021222}.
Let $W_q = \Xi_q ( K , \{ \ast \} ; R)$ where $\ast$ is the basepoint of $K$.
By definitions, $V(W_q ; A)$ is the tensor product $\bigotimes_{N_q} A$.
For a vector $w \in A$ and a linear functional $z$ on $A$, we define endomorphisms $\mathds{A}^w (c_{q+1})$ and $\mathds{B}^z (c_{q-1})$ as follows.
We use symbols $c_q$ to represent $q$-cells of $K$.
\begin{align}\notag
\mathds{A}^w (c_{q+1}) ( v ) &=  \bigotimes_{c_q} ( \phi ( [ c_{q+1} : c_q ] ) w^{(c_q)} ) v_{c_q} ,  \\
\mathds{B}^z (c_{q-1}) (v) &= z \left( \prod_{c_q} \left( \phi ([c_q : c_{q-1}]) v^{(2)}_{c_q} \right) \right) \cdot \bigotimes_{c_q} v^{(1)}_{c_q}  . \notag
\end{align}
Note that $\mathds{A}^{\sigma_A} (c_{q+1} ) = \mathds{S}^+ ( W_q , c_{q+1} ; A , \phi)$ and $\mathds{B}^{\sigma^A} (c_{q-1} ) = \mathds{S}^- ( W_q , c_{q-1} ; A , \phi)$ by definitions.
Here, $[ c : c^\prime ]$ denotes the incidence number between cells as before.
Especially they are integers so that to compute the $\phi$-action it suffices to refer to Remark \ref{202012282203}.
By Theorem \ref{202101021229} and Remark \ref{202101022117}, the eigenspaces of Hamiltonian $\sum_{c_{q+1}} (id - \mathds{A}^{\sigma_A} (c_{q+1}) ) + \sum_{c_{q-1}} ( id - \mathds{B}^{\sigma^A} (c_{q-1} ) )$ gives an orthogonal decomposition of $V(W_q ; A)$ with respect to an appropriate pairing and its 0-eigenspace is a homotopy invariant.
Moreover, the 0-eigenspace is naturally isomorphic to the ordinary homology Hopf algebra $H_q ( K , \ast ; A)$.

We have a family of LSM's parametrized by pointed finite CW-complex $F$ (see Example \ref{202007061120}).
Let $X = \Xi^F_q (K, \{ \ast \} ; R)$.
In the following subsubsections, we give a concrete description of the total space $V(X ; A)$ and the stabilizers by using those of $W_0, W_1 , W_2, \cdots$.

As a preparation, we define two endomorphisms $\mathds{A}^w_n (c_q )$ and $\mathds{B}^z_n ( c_q )$ on $V(W_q ; A)$ as follows.
Here, $n \in \mathbb{Z}$ and $v = \otimes_{c_q\in N_q} v_{c_q} \in V(X ; A)$.
\begin{align}\notag
\mathds{A}^w_n (c_q ) ( v) &= \left( ( \phi ( n ) w ) \cdot v_{c_q} \right) \otimes \left( \bigotimes_{c^\prime_q \neq c_q} v_{c^\prime_q} \right)  , \\
\mathds{B}^z_n ( c_q ) v &= z \left( \phi ( n )  v^{(2)}_{c_q} \right) \cdot \left( v^{(1)}_{c_q} \otimes \bigotimes_{c^\prime_q \neq c_q} v_{c^\prime_q} \right)  . \notag
\end{align}

For a $q$-cell $c_q$ of $K$ and a $r$-cell $d_r$ of $F$, we denote by $c_q \times d_r$ the product cell.
We note that the incidence numbers between product cells are computed as follows.
\begin{align}
[ c_q \times d_r : c_{q-1} \times d_r ] &= [ c_q : c_{q-1} ] , \notag \\
[ c_q \times d_r : c_q \times d_{r-1} ] &= (-1)^q [ d_r : d_{r-1} ] . \notag
\end{align}

\subsubsection{The case of $F=\ast$}

Let $F$ be a one-point space with the basepoint itself.
By definitions, we have $X \cong \Xi_q ( K , K ; R )$ which is the empty short abstract complex.
In this case, the associated vector space $V(X ; A) = k$, i.e. the ground field itself, and the stabilizers are trivial so that the elementary operator is zero.

\subsubsection{The case of $F = S^n$}

Let $F$ be the $n$-sphere $S^n = d_0 \cup d_n$ equipped with an $n$-cell $d_n$ and the basepoint $d_0$ (see the left side of Figure \ref{202101022151} for $n=1$).
By definitions, we have $X = \Xi_q ( K \times S^n , K \times \{ d_0 \} \cup \{ \ast_K \} \times S^n  ; R)$.
We have an obvious isomorphism $X \cong \Xi_{q-n} (K , \{ \ast_K \} ; R)$.
Hence, the total space $V (X ; A)$ is naturally isomorphic to $V ( W_{q-n} ; A)$.
Under the isomorphism, the $(\pm)$-stabilizers are computed as follows.
\begin{align}\notag
\mathds{S}^+ ( X , c_{q-n+1} \times d_n ; A , \phi ) &= \mathds{A}^{\sigma_A} (c_{q-n+1}) , \\
\mathds{S}^- ( X , c_{q-n-1} \times d_n  ; A  , \phi ) &= \mathds{B}^{\sigma^A} (c_{q-n-1} ) . \notag
\end{align}
Here, $c_q \times d_{q^\prime}$ is the $(q+q^\prime)$-cell of $K \times S^n$ defined by the product of $c_q , d_{q^\prime}$.
Hence, the elementary operator coincides with $\mathds{H} ( X ; A, \phi) = \sum_{c_{q-n+1}} ( id -\mathds{A}^{\sigma_A} (c_{q-n+1})) + \sum_{c_{q-n-1}} ( id - \mathds{B}^{\sigma^A} (c_{q-n-1} ))$.
By Theorem \ref{202101021229}, we see that the 0-eigenspace (equipped with the symmetric or Hermitian form) is isomorphic to the $q$-th ordinary homology Hopf algebra $H_q ( K \times S^n , K \times \{ d_0 \} \cup \{ \ast_K \} \times S^n ; A) \cong \widetilde{H}_q ( K \wedge S^n ; A) \cong \widetilde{H}_{q-n} ( K ; A)$ where $\wedge$ is the smash product and $\widetilde{H}_\bullet$ is the reduced homology theory.

We remark that for other CW-complex structures of $n$-sphere, $X \not\cong \Xi_{q-n} (  K , \{ \ast_K \} ; R)$ in general.
For example, consider $n=1$.
Let $S^1 = d_{0,0} \cup d_{0,1} \cup d_{1,0} \cup d_{1,1}$ equipped with 1-cells $d_{1,0}, d_{1,1}$ and 0-cells $d_{0,0}, d_{0,1}$ (see the right side of Figure \ref{202101022151}).
Let $d_{0,0}$ be the basepoint.
By definitions, we have a natural isomorphism $V(X ; A) \cong V( W_q ; A) \otimes V(Y^{q-1} ; A ) \otimes V( Y^{q-1} ; A)$ under which we have,
\begin{align}\notag
\mathds{S}^+ ( X , c_{q+1} \times d_{0,1} ; A , \phi) &= \mathds{A}^{\sigma_A} (c_{q+1} ) \otimes id \otimes id , \\
\mathds{S}^+ ( X , c_{q} \times d_{1,0} ; A , \phi) &= \mathds{A}^{\sigma^{(1)}_A}_{(-1)^q} (c_q ) \otimes \mathds{A}^{\sigma^{(2)}_A} (c_q ) \otimes id , \notag \\
\mathds{S}^+ ( X , c_{q} \times d_{1,1} ; A , \phi) &=\mathds{A}^{\sigma^{(1)}_A}_{(-1)^q} (c_q ) \otimes id \otimes \mathds{A}^{\sigma^{(2)}_A} (c_q ) , \notag \\
\mathds{S}^- ( X , c_{q-1} \times d_{0,1} ; A , \phi) &= \mathds{B}^{(\sigma^A)^{(1)}} (c_{q-1} ) \otimes \mathds{B}^{(\sigma^A)^{(2)}}_{(-1)^{q-1}} ( c_{q-1} ) \otimes \mathds{B}^{(\sigma^A)^{(3)}}_{(-1)^{q-1}} ( c_{q-1} ) , \notag \\
\mathds{S}^- ( X , c_{q-2} \times d_{1,0} ; A , \phi) &= id \otimes \mathds{B}^{\sigma^A} ( c_{q-2} ) \otimes id , \notag \\
\mathds{S}^- ( X , c_{q-2} \times d_{1,1} ; A , \phi) &= id \otimes id \otimes \mathds{B}^{\sigma^A} ( c_{q-2} ) . \notag
\end{align}
Here, $\Delta_{A^\vee} (\sigma^A) = (\sigma^A)^{(1)} \otimes (\sigma^A)^{(2)}$ and $(\Delta_{A^\vee} \otimes id ) \circ \Delta_{A^\vee}  (\sigma^A) = (\sigma^A)^{(1)} \otimes (\sigma^A)^{(2)} \otimes (\sigma^A)^{(3)}$ mean the Sweedler notation related with the dual Hopf algebra $A^\vee$.
Note that by Theorem \ref{202101021229} the 0-eigenspace of the elementary operator is isomorphic to $\widetilde{H}_{q-n} ( K ; A)$.

\begin{figure}[ht]
  \includegraphics[width=12cm]{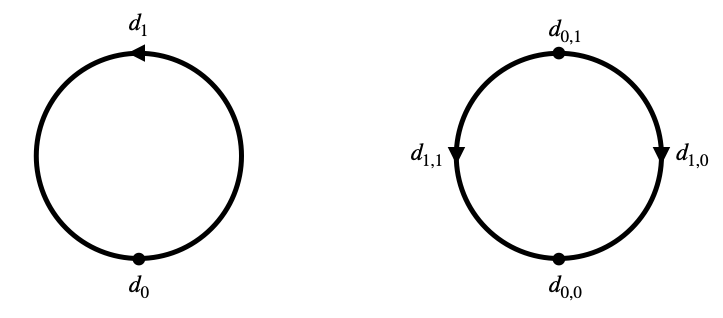}
  \caption{Some CW-complex structures of 1-sphere $S^1$}
  \label{202101022151}
\end{figure}

\subsubsection{The case of $F = Q_n$}

Recall the CW-complex $Q_n$ in Example \ref{202012282309}.
Consider the 0-cell $d_0$ as the basepoint of $Q_n$.
By definitions, we have $V(X ; A) \cong V(W_{q-1} ; A) \otimes V(W_{q-2} ; A)$.
In the physical sense, the qudit $A$ is assigned to $(q-1), (q-2)$-cells of $K$.
Under the isomorphism, the $(\pm)$-stabilizers are computed as follows.
In particular, the stabilizers are associated with $q, (q-1), (q-2), (q-3)$-cells of $K$.

\begin{align}\notag
\mathds{S}^+ ( X , c_q \times d_1 ; A , \phi) &= \mathds{A}^{\sigma_A} ( c_q ) \otimes id , \\
\mathds{S}^+ ( X , c_{q-1} \times d_2 ; A , \phi) &= \mathds{A}^{\sigma^{(1)}_A}_{(-1)^{q-1}n} (c_{q-1} ) \otimes \mathds{A}^{\sigma^{(2)}_A} ( c_{q-1} ) , \notag \\
\mathds{S}^- ( X , c_{q-2} \times d_1 ; A , \phi) &=\mathds{B}^{(\sigma^A)^{(1)}} (c_{q-2} ) \otimes \mathds{B}^{(\sigma^A)^{(2)}}_{(-1)^q n} (c_{q-2} ) ,\notag  \\
\mathds{S}^- ( X , c_{q-3} \times d_2 ; A , \phi) &=id \otimes \mathds{B}^{\sigma^A} (c_{q-3}) .\notag
\end{align}
By definitions, the elementary operator coincides with
\begin{align}\notag
\mathds{H} ( X ; A, \phi) =& 
\sum_{c_q} ( id - \mathds{A}^{\sigma_A} ( c_q ) \otimes id ) + \sum_{c_{q-1}} ( id - \mathds{A}^{\sigma^{(1)}_A}_{(-1)^{q-1}n} (c_{q-1} ) \otimes \mathds{A}^{\sigma^{(2)}_A} ( c_{q-1} ) ) \\
 &+ \sum_{c_{q-2}} ( id - \mathds{B}^{(\sigma^A)^{(1)}} (c_{q-2} ) \otimes \mathds{B}^{(\sigma^A)^{(2)}}_{(-1)^q n} (c_{q-2} ) ) + \sum_{c_{q-3}} ( id- id \otimes \mathds{B}^{\sigma^A} (c_{q-3}) ) . \notag
\end{align}
By Theorem \ref{202101021229} and Remark \ref{202101022117}, we see that the 0-eigenspace equipped with the symmetric bilinear or Hermitian form is isomorphic to the $q$-th ordinary homology Hopf algebra $H_q ( K \times Q_n , K \times \{ d_0 \} \cup \{ \ast_K \} \times Q_n ; A) \cong \widetilde{H}_q ( K \wedge Q_n ; A)$ of the smash product of $K$ and $Q_n$.

\begin{Example}
For the convenience of the readers, we give an application to group Hopf algebra $A = kG$..
Recall that the normalized integral and cointegral are given by $\sigma_A = |G|^{-1} \sum_{g \in G} g$ and $\sigma^A = \hat{e}$.
Let $\phi$ be any $R$-action on $A$.
Then we obtain
\begin{align}\notag
\mathds{S}^+ ( X , c_q \times d_1 ; A , \phi) &= |G|^{-1}  \sum_{g\in G} \mathds{A}^{g} ( c_q ) \otimes id , \\
\mathds{S}^+ ( X , c_{q-1} \times d_2 ; A , \phi) &= |G|^{-1} \sum_{g \in G} \mathds{A}^{g}_{(-1)^{q-1}n} (c_{q-1} ) \otimes \mathds{A}^{g} ( c_{q-1} ) , \notag \\
\mathds{S}^- ( X , c_{q-2} \times d_1 ; A , \phi) &= \sum_{g \in G} \mathds{B}^{\hat{g}} (c_{q-2} ) \otimes \mathds{B}^{\hat{g^{-1}}}_{(-1)^q n} (c_{q-2} ) ,\notag  \\
\mathds{S}^- ( X , c_{q-3} \times d_2 ; A , \phi) &=id \otimes \mathds{B}^{\hat{e}} (c_{q-3}) .\notag
\end{align}
\end{Example}

\subsection{Reproduction of Kitaev's local stabilizers}
\label{201907220232}

In this subsection, we apply the results in subsection \ref{201907220230} to polyhedral complex which is a special class of CW-complex.
We interpret the bicommutative Kitaev's model as a part of a contravariant LSM in Example \ref{202101062347}.
Let $K$ be a pointed polyhedral complex with dimension lower than $2$ (see Definition 2.39 \cite{Koz}).
Let $X = \Xi^1 ( K , \{ \ast \} ; \mathbb{Z} )$ where $\ast$ is the basepoint.

\begin{remark}
For a polyhedral surface $K^\prime$ without a basepoint (which is the usual setting for the toric code), we set $K = K^\prime\amalg\{\ast\}$ by adding a disjoint point $\ast$ and consider it as a basepoint.
\end{remark}

Let $A$ be a finite-dimensional bisemisimple bicommutative Hopf algebra over $k$.
Denote by $\phi$ the canonical $\mathbb{Z}$-action.
See Remark \ref{202012282203}.

\begin{figure}[ht]
  \includegraphics[width=14.0cm]{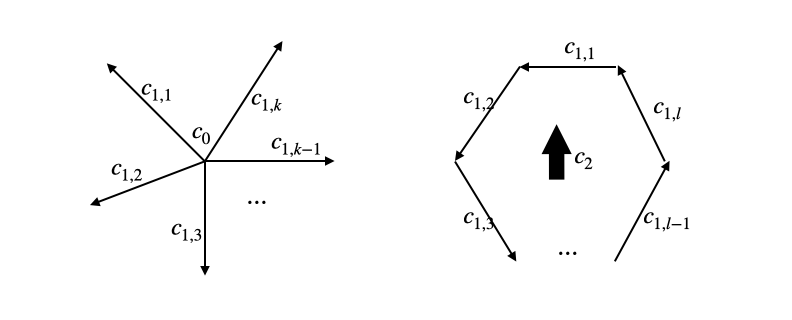}
  \caption{}
  \label{201907171740}
\end{figure}

Let $c_0  \in X_+$, i.e. $c_0$ a $0$-cell which is not a basepoint.
We compute the operator $\mathds{S}^+ ( X , c_0 ; A , \phi)$.
Suppose that edges on $K$ containing $c_0$ are $c_{1,1}, c_{1,2} , \cdots c_{1,k}$.
Then we have $[c_{1,r} : c_0]$ is $1$ or $(-1)$ due to Proposition \ref{201907171754}.
For example if $c_0$ is the source of $c_{1,r}$'s as Figure \ref{201907171740}, then $[c_{1,r} : c_0] =1$ due to the relations of their orientations.
By definitions, we have
\begin{align}\notag
\mathds{S}^+ ( X , c_0 ; A , \phi) ( \bigotimes_{c_1} v_{c_1} ) = \bigotimes_{c_1} \left( \sigma^{(c_1)}_A \cdot v_{c_1} \right) .
\end{align}

Analogously, we compute the operator $\mathds{S}^- ( X , c_2 ; A , \phi)$ for $c_2 \in X_-$, i.e. $c_2$ a 2-cell.
Suppose that faces of $c_2$ are $c_{1,1}, c_{1,2} ,\cdots , c_{1,l}$.
Then we have $[c_2 : c_{1,r} ]$ is $1$ or $(-1)$ due to Proposition \ref{201907171754}.
For example, if the orientation of $c_{1,r}$ coincides with that of the boundary of $c_2$ as Figure \ref{201907171740}, then $[c_2 : c_{1,r} ] =1$.
Hence, we obtain
\begin{align}\notag
\mathds{S}^- ( X , c_2 ; A , \phi ) ( \bigotimes_{c_1} v_{c_1} ) = \sigma^A ( \prod_{c_1} v^{(2)}_{c_1} ) \cdot \bigotimes_{c_1} v^{(1)}_{c_1} .
\end{align}

In the literature, the stabilizers on surfaces are constructed from a semisimple (not necessarily bicommutative) Hopf algebra over some field with characteristic zero.
In \cite{Kit}, the stabilizers are defined based on group (Hopf) algebras over $\mathbb{C}$.
Consider $k = \mathbb{C}$ and $A = \mathbb{C} G$ where $G$ is a finite group.
Then the Kitaev's stabilizers $\mathbf{A} (c_0)$ and $\mathbf{B} (c_2)$ are defined.
Here, we use the bold font to avoid confusion with the Hopf algebra $A$.
If $G$ is abelian, then we have
\begin{align}\notag
\mathbf{A} (c_0) = \mathds{S}^+ ( X , c_0 ; \mathbb{C}G , \phi ),~~ \mathbf{B}(c_2) = \mathds{S}^- ( X , c_2 ; \mathbb{C}G , \phi ) .
\end{align}
In \cite{meusburger}, the stabilizers are constructed from a finite-dimensional semisimple Hopf algebra $A$ over a field with characteristic zero.
The stabilizers are given by $\mathbf{A}^{\sigma_A}_{c_0}$ and $\mathbf{B}^{\sigma^A}_{c_2}$.
Note that $A$ is bisemisimple since the ground field $k$ is of characteristic zero.
If $A$ is bicommutative, then we have
\begin{align}\notag
\mathbf{A}^{\sigma_A}_{c_0} = \mathds{S}^+ ( X , c_0 ; A , \phi ), ~~ \mathbf{B}^{\sigma^A}_{c_2} = \mathds{S}^- ( X , c_2 ; A , \phi ) .
\end{align}
In \cite{BMCA}, Hopf C$^\ast$-algebra is the main input.
A bicommutative bisemisimple Hopf C$^\ast$-algebra should be the function Hopf C$^\ast$-algebra of a finite abelian group.
Hence, this case is implied by the previous one.
Similarly, the stabilizers in \cite{BalKir} are realized implicitly by the previous case since the ground field is $\mathbb{C}$.


\subsection{Extension of 0-eigenspaces to TQFT}
\label{202101021342}

We give a generalization of the relationship between Kitaev model and Turaev-Viro TQFT based on bicommutative settings : we prove that 0-eigenspaces associated with a topological LSM extends to a projective TQFT.
We make use of some integrals along Hopf homomorphisms to construct the morphisms corresponding to cobordisms.
The results in this subsection heavily depend on our other papers \cite{kim2019integrals} \cite{kim2020family}.

In \cite{kim2019integrals}, we introduce a new notion of integral along bimonoid homomorphisms in a symmetric monoidal category.
We give an overview by restricting ourselves to the tensor category vector spaces.

\begin{Defn}
Let $A,B$ be Hopf algebras and $\xi : A\to B$ be a Hopf homomorphism.
A linear homomorphism $\mu : B \to A$ is a {\it right integral along} $\xi$ if $\mu (b) \cdot a = \mu ( b \cdot \xi (a ) )$ and $\mu ( b^{(1)} ) \otimes b^{(2)} = \mu (b)^{(1)} \otimes \xi ( \mu (b)^{(2)} )$.
A linear homomorphism $\mu : B \to A$ in $\mathcal{C}$ is a {\it left integral along} $\xi$ if $a \cdot \mu (b) = \mu ( \xi (a) \cdot b )$ and $b^{(1)} \otimes \mu (b^{(2)} ) = \xi ( (\mu (b))^{(1)} ) \otimes \mu (b)^{(2)}$.
A linear homomorphism $\mu : B \to A$ is an {\it integral along} $\xi$ if it is a right integral along $\xi$ and a left integral along $\xi$.
An integral $\mu$ along $\xi$ is {\it normalized} if $\xi ( \mu ( \xi ( a ) ) ) = \xi ( a )$.
An integral $\mu$ along $\xi$ is {\it a generator} if $\mu^\prime ( \xi ( \mu ( b ) ) ) = \mu ( b ) = \mu ( \xi ( \mu^\prime ( b ) ) )$ for any integral $\mu^\prime$ along $\xi$.
\end{Defn}

\begin{prop}
\label{202101042121}
Let $A,B$ be finite-dimensional bisemisimple bicommutative Hopf algebras.
For any Hopf homomorphism $\xi : A \to B$, there exists a unique normalized generator integral $\mu_\xi$ along $\xi$.
\end{prop}
\begin{proof}
A finite-dimensional bisemisimple bicommutative Hopf algebra has a finite volume in the sense of \cite{kim2019integrals}.
In fact, Proposition \ref{202101041125} implies the existence of a normalized integral and cointegral.
Moreover, $\mathrm{vol}^{-1} (A) = \sigma^A \circ \sigma_A$ is invertible by \cite{BKLT}.
Then by the main theorem of \cite{kim2019integrals} we obtain the result.
\end{proof}

Fix a (not necessarily topological) LSM $\Xi$ over $R$.
Let $A$ be a finite-dimensional bisemisimple bicommutative Hopf algebra with an $R$-action $\phi$.
For a pointed finite CW-complex $K$, let $E (K) = H( \Xi ( K ) ; A , \phi )$.
If $K$ is a pointed subcomplex of $L$, say that the embedding is $i : K \hookrightarrow L$, then we have an inclusion $\Xi (i) : \Xi (K ) \rightarrowtail \Xi ( L )$ which induces a Hopf homomorphism $E (i) : E (K) \to E (L)$.

\begin{remark}
For a contravariant LSM $\Xi^\prime$, we consider $E (K) = H( \Xi^\prime (K) ; A , \phi)$.
An inclusion $i : K \hookrightarrow L$ induces a restriction $\Xi^\prime (i) : \Xi^\prime (L) \twoheadrightarrow \Xi^\prime (K)$ which contravariantly induces $E(i): E(L) \to E(K)$ by Proposition \ref{202101072053}.
The discussion below could be applied to such a contravariant functor $E$ with appropriate modifications.
\end{remark}

\begin{prop}
\label{202101042120}
The Hopf algebra $E(K)$ is a finite-dimensional bisemisimple bicommutative Hopf algebra.
\end{prop}
\begin{proof}
Let $\mathsf{Hopf}^\mathsf{bc,vol}_k$ be the category of finite-dimensional bisemisimple bicommutative Hopf algebras (equivalently, bicommutative Hopf algebras with a finite volume) and Hopf homomorphisms.
The category $\mathsf{Hopf}^\mathsf{bc,vol}_k$ is a abelian subcategory of $\mathsf{Hopf}^\mathsf{bc}_k$ \cite{kim2019integrals}.
Thus, $E(K)$ is an object of $\mathsf{Hopf}^\mathsf{bc,vol}_k$ since it is defined by a homology Hopf algebra of a chain complex in the abelian category $\mathsf{Hopf}^\mathsf{bc,vol}_k$.
\end{proof}

From now on, we construct a path-integral of $E$ along a cobordism by using the integrals along Hopf homomorphisms.
More generally, we consider a cospan of embeddings instead of cobordisms, which is technically useful.

\begin{Defn}
\label{202101022139}
{\it A cospan of embeddings} is a quintuple $\Lambda = (L ; K_0 , f_0 ; K_1 , f_1)$ where $L,K_0,K_1$ are pointed finite CW-complexes and $f_0 : K_0 \hookrightarrow L$, $f_1 : K_1 \hookrightarrow L$ are embeddings.
For a cospan $\Lambda = (L ; K_0 , f_0 ; K_1 , f_1)$, we call $f_0, f_1$ {\it the left arm and the right arm} respectively.

Consider two cospans of embeddings $\Lambda_0 = (L ; K_0 , f_0 ; K_1 , f_1)$ and $\Lambda_1 = (L^\prime ; K_0 , f^\prime_0 ; K_1 , f^\prime_1)$ with the same source $K_0$ and the target $K_1$.
The cospans $\Lambda_0, \Lambda_1$ are {\it homotopy equivalent} if there exists a pointed homotopy equivalence $h : L \to L^\prime$ such that $h \circ f_0 \simeq f^\prime_0$ and $h \circ f_1 \simeq f^\prime_1$.
We write $\Lambda_0 \simeq \Lambda_1$.
\end{Defn}

\begin{Defn}
\label{202101030308}
Consider a cospan of embeddings $\Lambda = (L ; K_0 , f_0 ; K_1 , f_1)$.
We define {\it a path-integral along $\Lambda$} by $(\hat{\mathsf{PI}} (E)) ( \Lambda ) \stackrel{\mathrm{def.}}{=} \mu_{E(f_1)} \circ E(f_0)$ which is a linear homomorphism from $E(K_0)$ to $E(K_1)$.
Note that $\mu_{E(f_1)}$ exists due to Proposition \ref{202101042121}, \ref{202101042120}.
\end{Defn}

For simplicity, let $Y = \hat{\mathsf{PI}} (E)$.
In the following propositions, we give some basic properties.

\begin{prop}
For an embedding $f : K \hookrightarrow L$, let $\iota ( f ) = ( L ; K , f ; L , id_L)$.
Then we have $Y( \iota ( f) ) = E (f)$.
\end{prop}
\begin{proof}
Note that for a Hopf isomorphism $\xi$, the normalized generator integral $\mu_\xi$ should be the inverse of $\xi$.
Then the claim follows from the definition.
\end{proof}

\begin{prop}
\label{202012280910}
Let $\Lambda_1 = (L^\prime ; K_1 , g_1 ; K_2 , g_2)$, $\Lambda_0 = (L ; K_0 , f_0 ; K_1 , f_1)$ be cospans of embeddings of pointed finite CW-complexes and $\Lambda_1 \circ \Lambda_0$ be the composition.
In other words, $\Lambda_1 \circ \Lambda_0 = ( L^\prime \vee_{K_1} L ; K_0 , j_0 \circ f_0 ; K_2 , j_2 \circ f_2)$ where $L^\prime \vee_{K_1} L$ is the finite CW-complex obtained by gluing $L,L^\prime$ along the subcomplex $K_1$ and $j_0 , j_2$ are the embeddings from $L,L^\prime$ into $L^\prime \vee_{K_1} L$ respectively.
Then there exists a unique $0 \neq \lambda \in k$ such that
\begin{align}\notag
Y ( \Lambda_1 ) \circ Y( \Lambda_0 ) = \lambda \cdot Y ( \Lambda_1 \circ \Lambda_0 ) .
\end{align}
Furthermore, if $E(g_2)$ is an epimorphism or $E(j_2)$ is a monomorphism in $\mathsf{Hopf}^\mathsf{bc}_k$, then $\lambda = 1$.
\end{prop}
\begin{proof}
By Theorem 12.1 in \cite{kim2019integrals}, we have $Y ( \Lambda_1 ) \circ Y( \Lambda_0 ) = \lambda \cdot Y ( \Lambda_1 \circ \Lambda_0 )$ for $\lambda = \langle \mathrm{cok_H} ( E(g_2) ) \circ \mathrm{ker_H} ( E(j_2) ) \rangle$ where $\langle - \rangle$ denotes the inverse volume of Hopf homomorphisms.
Especially, if $E(g_2)$ is an epimorphism or $E(j_2)$ is a monomorphism in $\mathsf{Hopf}^\mathsf{bc}_k$, then $\langle \mathrm{cok_H} ( E(g_2) ) \circ \mathrm{ker_H} ( E(j_2) ) \rangle$ where $\langle - \rangle$ should be $1 \in k$ by definition.
\end{proof}

Note that the homomorphism $E(i)$ is only defined for an embedding $i$.
If $\Xi$ is topological, then it extends to any pointed maps between finite pointed CW-complexes.
As a preparation, we introduce the following cospan of embeddings induced by pointed maps.

\begin{Defn}
Let $f : K \to L$ be a pointed map which is not necessarily an embedding.
We define a cospan of embeddings $J( f)$ by $( M(f) ; L , K ; L ; r )$ where $M(f) = ( K \wedge [0,1]^+ ) \vee_{K} L$ denotes the mapping cylinder of $f$ obtained by identifying $(x,0) \sim x$ for $x \in K$.
$l,r$ denote the canonical embeddings.
\end{Defn}

\begin{Lemma}
\label{202012291433}
Let $f_0, f_1 : K \hookrightarrow L$ be embeddings which are isotopic to each other, i.e. there exists an embedding $f : K \wedge [0,1]^+ \hookrightarrow L$ such that $f\circ i_0 = f_0$ and $f \circ i_1 = f_1$.
If the LSM $\Xi$ is topological with respect to $(A, \phi)$, then $E (f_0 ) = E( f_1 )$.
\end{Lemma}
\begin{proof}
Note that $E( f_t ) = E( f \circ i_t) = E(f) \circ E(i_t)$ by the functoriality of $E$.
Since $\Xi$ is topological, we have $E (i_0 ) = E(i_1 )$ so that we obtain $E( f_0 ) = E( f_1 )$.
\end{proof}

\begin{Defn}
{\it A $\mathsf{Hopf}^\mathsf{bc}_k$-valued (covariant) Brown functor $E^\prime$} is a symmetric monoidal functor from the category of pointed finite CW-complexes and pointed maps to $\mathsf{Hopf}^\mathsf{bc}_k$.
It satisfies the following axioms.
\begin{itemize}
\item
It is a homotopy invariant :
A pointed homotopy of maps $f_0 \simeq f_1$ implies $E^\prime (f_0 ) = E^\prime (f_1)$.
\item
It satisfies the Mayer-Vietoris axiom :
If $K_0,K_1$ are subcomplexes of $L$, then the following induced chain complex of bicommutative Hopf algebras is
\begin{align}\notag
\notag
E^\prime ( K_0 \cap K_1 ) \to E^\prime ( K_0 ) \otimes E^\prime ( K_1 )  \to E^\prime ( K_0 \cup K_1 )  .
\end{align}
\end{itemize}
\end{Defn}

\begin{Lemma}
\label{202012291530}
If the LSM $\Xi$ is topological with respect to $(A, \phi)$, then the following assignments give a $\mathsf{Hopf}^\mathsf{bc}_k$-valued (covariant) Brown functor.
\begin{itemize}
\item
It assigns the bicommutative Hopf algebra $E(K)$ to each pointed finite CW-complex $K$.
\item
It assigns the linear homomorphism $Y(J(f)) : E(K) \to E(L)$ to each pointed map $f : K \to L$.
\end{itemize}
\end{Lemma}
\begin{proof}
Let $f : K \hookrightarrow L$ be an embedding.
Firstly we show that $Y ( J(f)) = E(f)$.
Let $r, l$ be the left and right arms of the cospan $J(f)$.
Then $r \circ f$ is isotopic to $l$.
By Lemma \ref{202012291433} and the functoriality, we obtain $E(r) \circ E(f) = E(r \circ f) = E(l)$.
Hence, $Y( J(f)) = \mu_{E(r)} \circ E(l) = \mu_{E(r)} \circ E (r) \circ E(f)$.
Here, $E(r)$ is an isomorphism since $\Xi$ is topological and $r$ is a homotopy equivalence.
It implies that $\mu_{E(r)} = E(r)^{-1}$ so that $Y ( J(f)) = E(f)$.

In particular, we also obtain $Y(J( Id_K ) ) = Id_{E(K)}$.

Let $f: K_0 \hookrightarrow K_1$ be an embedding and $g : K_1 \to K_2$ be an arbitrary pointed map.
We prove that $Y ( J(g)) \circ Y ( J(f)) = Y ( J(g\circ f))$.
By Proposition \ref{202012280910}, it suffices to prove that $Y( J(g) \circ J(f) ) = Y( J(g \circ f) )$.
We set an embedding $h : M( g\circ f) \hookrightarrow M(g) \vee_{K_1} M(f)$ by $h( [x , t ] ) = [ f(x) , t ] \in K_1 \wedge [0,1]^+ \subset M(g) \vee_{K_1} M(f)$ if $(x, t) \in K_0 \times ( 0,1 ]$ and $h( y) = y$ if $y \in K_2 \subset M(g) \vee_{K_1} M(f)$.
Let $r,l$ be the right and left arms of $J(g \circ f)$ and $r^\prime, l^\prime$ be those of $J(g) \circ J(f)$.
By definitions, we have $h \circ r = r^\prime$.
Furthermore, $h \circ l$ and $l^\prime$ are isotopic to each other.
Hence, we obtain $E(h) \circ E(r) = E(r^\prime)$ and $E(h) \circ E(l) = E(l^\prime)$ by Lemma \ref{202012291433}.
It gives $Y( J(g) \circ J(f) ) = \mu_{E(r^\prime)} \circ E(l^\prime ) = \mu_{E(h) \circ E(r) } \circ E(h) \circ E(l)$.
Note that $\mu_{E(h) \circ E(r) } = \mu_{E(r)} \circ \mu_{E(h)}$ since $E(r)$ is an isomorphism (due to the topologicality of $\Xi$).
Since $h$ is a homotopy equivalence and $\Xi$ is topological, $\Xi (h)$ is an isomorphism so that $\mu_{E(h)} = E(h)^{-1}$.
Thus, $Y( J(g) \circ J(f) ) = Y( J(g \circ f) )$.

Let $f_0, f_1 : K \to L$ be arbitrary pointed maps.
We show that if $f_0,f_1$ are homotopic to each other then $Y ( J(f_0 ) ) = Y( J(f_1 ) )$.
In fact, if $f : K \wedge [0,1]^+ \to L$ is a homotopy from $f_0$ to $f_1$, then we have $Y ( J(f_t ) ) = Y( J( f \circ i_t) ) = Y( J(f)) \circ Y(J(i_t))$ due to the above result.
Moreover $Y(J(i_t)) = E(i_t)$ since $i_t$ is an embedding.
We obtain $Y ( J(f_0 ) ) = Y( J(f_1 ) )$ since $E(i_0) = E(i_1)$.

We prove that $Y(J(-))$ preserves compositions.
Let $f : K_0 \to K_1$ and $g : K_1 \to K_2$ be arbitrary pointed maps.
Denote by $l,r$ the left and right arms of the cospan $J(f)$.
Note that there exists a homotopy inverse $h : M(f) \to K_1$ of $r$.
If we choose such $h$, then we have $Y(J(g)) = Y(J( g \circ h \circ r)) = Y(J(g \circ h)) \circ Y(J(r))$ and $Y(J(f)) = Y( J(h \circ l )) = Y(J(h)) \circ Y(J(l))$ due to the above results.
Thus, $Y(J(g)) \circ Y(J(f)) = Y(J(g \circ h)) \circ Y(J(r)) \circ Y(J(h)) \circ Y(J(l))$.
We claim that $Y(J(r)) \circ Y(J(h))$ is the identity.
In fact, $Y(J(h)) \circ Y(J(r)) = Y(J(h \circ r)) = Y(J( Id_{K_1})) = Id_{E(K_1)}$ since $r$ is an embedding.
Moreover, we have $Y(J(h)) = Y(J(r))^{-1}$ since $Y(J(r)) = E(r)$ is an isomorphism (note that $r$ is a homotopy equivalence).
By the claim, we prove $Y(J(g)) \circ Y(J(f)) = Y(J(g \circ h)) \circ Y(J(l)) = Y( J( g\circ h \circ l)) = Y( J( g \circ f ) )$.
Above all, it is proved that the assignments $E(K)$ and $Y(J(f))$ give a functor from the homotopy category of pointed finite CW-complexes and pointed maps to $\mathsf{Hopf}^\mathsf{bc}_k$.
Finally, the functor is improved to a symmetric monoidal functor in the obvious way due to the monoidal axiom of the LSM $\Xi$.
Furthermore, the functor satisfies the Mayer-Vietoris axiom due to Proposition \ref{202101021609}.
It completes the proof.
\end{proof}

\begin{Example}
For the LSM in Example \ref{202007061120}, we have $E(K) \cong \widetilde{H}_q ( K  ; A )$ where $\widetilde{H}_\bullet ( - ; A)$ is the reduced ordinary homology theory.
Similarly, for the LSM in Example \ref{202101021222}, we have $E(K) \cong H_q ( K \times F , \{ \ast \} \times F \cup K \times \{ \ast \} ; A ) \cong \widetilde{H}_q ( K \wedge F ; A)$.
\end{Example}

By combining Lemma \ref{202012291530} with our result in \cite{kim2020family}, we obtain the following theorem.

\begin{theorem}
\label{202101042109}
Let $n \in \mathbb{N}$.
Suppose that the LSM $\Xi$ is topological with respect to $(A , \phi)$.
\begin{itemize}
\item
There exists a projective $n$-TQFT $Z$ such that $Z (M^{n-1} ) = V_0 ( \Xi ( M^{n-1} ) ; A , \phi)$ for a closed $(n-1)$-manifold $M^{n-1}$.
\item
If $\Xi$ is the LSM in Example \ref{202007061120}, then there exists an $n$-TQFT $Z$ such that $Z (M^{n-1} ) = V_0 ( \Xi ( M^{n-1} ) ; A , \phi)$ and $Z ( N^n ) = \prod_{r \geq 0} \dim H_{q-r} ( N \times F , N \times \ast ; A )^{(-1)^r}$ for a closed $n$-manifold $N^n$.
\end{itemize}
\end{theorem}
\begin{proof}
We sketch the proof.
Let $E^\prime$ be the Brown functor in Lemma \ref{202012291530}.
For a cobordism $N^n$ from $M^{n-1}_0$ to $M^{n-1}_1$, define a cospan by $\Phi ( N^n ) = ( N^n ; f_0 , M^{n-1}_0 ; f_1 , M^{n-1}_0)$ 
Let $Z(N^n) = \mu_{E^\prime(f_1)} \circ E^\prime (f_0)$ where $f_0 : M^{n-1}_0 \to N^n,f_1 : M^{n-1}_1 \to N^n$ are inclusions.
Then the first claim basically follows from Proposition \ref{202012280910}.

As explained in Proposition \ref{202012280910}, the composition is preserved up to a scalar in the ground field $k$.
Such scalars assigned to composable cobordisms $\Lambda_0, \Lambda_1$ induce a 2-cocycle of the cobordism category, which is an analogue of the 2-cocycle induced by a projective action of a group.
The 2-cocycle is a coboundary of a 1-cochain if $\Xi$ is the LSM in Example \ref{202007061120} since the eigenspace is isomorphic to the $q$-th homology theory $H_q ( (-) \wedge F ; A)$ which is bounded below.
In fact, for a $\mathsf{Hopf}^\mathsf{bc,vol}_k$-valued homology theory which is bounded below, the 2-cocycle is a coboundary of a canonical 1-cochain \cite{kim2020family}.
The symbol $\mathsf{Hopf}^\mathsf{bc,vol}_k$ is explained in the proof of Proposition \ref{202101042120}.
By using the 1-cochain, we could construct the $n$-TQFT in the statement from the previous projective $n$-TQFT.
\end{proof}

\begin{remark}
There is a homotopy-theoretic analogue of cobordism categories, called the cospan category $\mathsf{Cosp} ( \mathsf{CW}^\mathsf{fin}_\ast )$ of pointed finite CW-spaces \cite{kim2020family}.
We have an improvement of Theorem \ref{202101042109} : there exists a symmetric monoidal projective functor $Z : \mathsf{Cosp} ( \mathsf{CW}^\mathsf{fin}_\ast ) \to \mathsf{Vec}_k$ such that $Z(K) = V_0 ( \Xi ( K ) ; A , \phi)$.
For the LSM in Example \ref{202007061120}, we have an analogous result.
The proof is parallel with the above one.
\end{remark}

In Remark \ref{202101030318}, we give a question whether the definition of topological LSM's is appropriate.
Before we close this section, we give an answer as an application of Lemma \ref{202012291530}.
By the following proposition, our definition is essential to extend the 0-eigenspaces $E(K)$ with the assignment $Y(J(f))$ to a Brown functor ; to consider the path-integral $Y(\Lambda)$ as a homotopy invariant.

\begin{prop}
\label{202101022059}
The following propositions are equivalent with each other.
\begin{itemize}
\item
The LSM $\Xi$ is topological with respect to $(A, \phi)$.
\item
The assignments $E(K)$ and $Y(J(f))$ give a $\mathsf{Hopf}^\mathsf{bc}_k$-valued (covariant) Brown functor.
\item
The assignment $Y (\Lambda)$ is a homotopy invariant of $\Lambda$, i.e. $\Lambda_0 \simeq \Lambda_1$ implies $Y ( \Lambda_0 ) = Y ( \Lambda_1 )$.
\end{itemize}
\end{prop}
\begin{proof}
By Lemma \ref{202012291530}, the first part implies the second part.

We prove that the second part implies the third part.
We assume that the assignments $E(K)$ and $Y(J(f))$ give a $\mathsf{Hopf}^\mathsf{bc}_k$-valued (covariant) Brown functor $E^\prime$.
Let $\Lambda_0 = (L ; K_0 , f_0 ; K_1 , f_1 )$ and $\Lambda_1 = ( L^\prime ; K_0 , g_0 ; K_1, g_1)$ be cospans of embeddings between pointed finite CW-complexes.
Suppose that $h : L \to L^\prime$ is a pointed homotopy equivalence such that $f_t \circ h \simeq g_t$ for $t = 0, 1$.
Then we have $E^\prime ( h ) \circ E^\prime ( f_t ) = E^\prime ( h \circ f_t ) = E^\prime (g_t)$.
Hence, $Y( \Lambda_1 ) = \mu_{E(g_1)} \circ E(g_0) = \mu_{E^\prime (g_1)} \circ E^\prime(g_0) = \mu_{E^\prime (f_1)} \circ \mu_{E^\prime (h)} \circ E^\prime ( h ) \circ E^\prime ( f_0)$.
Note that $E^\prime (h)$ is an isomorphism since $h$ is a homotopy equivalence and $E^\prime$ preserves homotopy.
Thus, $\mu_{E^\prime (h)} = E^\prime (h)^{-1}$ so that $Y( \Lambda_1 ) = \mu_{E^\prime (f_1)}  \circ E^\prime ( f_0) = Y( \Lambda_0 )$.
It proves that $Y(\Lambda)$ is a homotopy invariant of $\Lambda$.

In the final step, we prove that the third part implies the first part.
Let us assume that $Y(\Lambda)$ is a homotopy invariant of $\Lambda$.
Let $K$ be a pointed finite CW-complex.
Consider a cospan $\Lambda_t = ( K \wedge [0,1]^+ ; K , i_t ; id , K \wedge [0,1]^+ )$ for $t = 0,1$.
Note that $\Lambda_0, \Lambda_1$ are homotopy equivalent with each other.
Hence, $Y(\Lambda_0) = Y( \Lambda_1)$ which implies $E(i_0) = E(i_1)$ by definitions.
It proves that $\Xi$ satisfies the first part of Definition \ref{202012291544}.
We prove the second part of Definition \ref{202012291544}.
Let $f : K \hookrightarrow L$ be an embedding which is a homotopy equivalence.
Choose a homotopy inverse $g$ of $f$ (which needs not be an embedding).
We have $Y( J(g) ) \circ Y ( \iota (f)) = Y( J(g) \circ \iota (f))$ by Proposition \ref{202012280910}.
Note that $\Lambda \circ \iota (f)$ is homotopy equivalent with the identity cospan $\iota (id_K)$ of $K$ since $g \circ f \simeq id_K$.
Hence, $Y( J(g) \circ \iota (f)) = Y( \iota (id_K)) = id_{E(K)}$ so that $Y( J(g) ) \circ E(f) = Y( J(g)) \circ Y( \iota (f)) = id_{E(K)}$.
It implies that $E(f)$ is a monomorphism in the category $\mathsf{Hopf}^\mathsf{bc}_k$.
Analogously, $f \circ g \simeq id_L$ implies $E(f) \circ Y(J(g)) = id_{E(L)}$ so that $E(f)$ is an epimorphism $\mathsf{Hopf}^\mathsf{bc}_k$.
Thus, $E(f)$ is an isomorphism since $\mathsf{Hopf}^\mathsf{bc}_k$ is an abelian category.
It completes the proof.
\end{proof}

\section{Poincar\'e-Lefschetz duality of local stabilizers}
\label{201907220234}

In this section, we propose a Poincar\'e-Lefschetz duality of the $(\pm)$-stabilizers in framework.
Let $M$ be a polyhedral complex.
For a subcomplex $K$ of $M$, denote by $K^\star$ the complex consisting of dual cells in the supplement of $K$ in $M$ (see Definition \ref{202007062100}).
In general, $K^\star$ is a regular cell-complex with oriented homology cells, which may not be a polyhedral complex.
Note that the notations and results in Example \ref{202012261747}, \ref{202012272320} formally extends to regular cell-complexes with oriented homology cells.

\begin{theorem}[Poincar\'e-Lefschetz Duality]
\label{201907200015}
Let $R$ be a commutative unital ring.
Let $M$ be a closed $R$-oriented $m$-manifold $M$ with a polyhedral complex structure.
Let $L \subset K$ be its polyhedral subcomplexes.
Then we have an isomorphism of short abstract complexes over $R$ in the sense of Definition \ref{202007080920},
\begin{align}
\label{201907260947}
\Xi_q ( L^\star ,  K^\star ; R) \cong
\Xi_{m-q} ( L , K ; R )^T   . 
\end{align}
\end{theorem}
\begin{proof}
The isomorphism follows from the Poincar\'e-Lefschetz Duality of cellular (co)chain complexes in Corollary \ref{201907211913}.
\end{proof}

\begin{Corollary}
\label{201907221136}
Recall the assumptions in Theorem \ref{201907200015}.
Let $X = \Xi_{m-q} ( L , K ; R)$ and $Y = \Xi_q ( L^\star , K^\star ; R)$.
\begin{enumerate}
\item
For an $(m-q-1)$-cell $c_-$ and its dual cell $(q+1)$-cell $c^\vee_-$, we have
\begin{align}\notag
\mathds{S}^+ (Y,  c^\vee_- ; A , \phi) &\cong \mathds{S}^- (X, c_- ; A^\vee , \phi^\vee )^\vee .
\end{align}
\item
For an $(m-q+1)$-cell $c_+$ and its dual cell $(q-1)$-cell $c^\vee_+$, we have
\begin{align}\notag
\mathds{S}^- ( Y, c^\vee_+ ; A , \phi ) &\cong \mathds{S}^+ ( X , c_+ ; A^\vee , \phi^\vee )^\vee .
\end{align}
\item
We have
\begin{align}\notag
\mathds{H} ( Y ; A , \phi ) &\cong \mathds{H} ( X ; A^\vee , \phi^\vee )^\vee .
\end{align}
\end{enumerate}
\end{Corollary}
\begin{proof}
It follows from Proposition \ref{202006302146} and Theorem \ref{201907200015}.
\end{proof}

\begin{Example}
We reformulate the duality in \cite{BraKit}.
Let $K$ be a polyhedral surface in a closed $\mathbb{Z}$-oriented 2-manifold.
If we denote by the dual complex by $K^\vee$, then we have an excision isomorphism $\Xi^1 ( K^\vee , (\partial_z K)^\vee ; \mathbb{Z} ) \cong \Xi^1 ( (\partial_x K)^\star , K^\star ; \mathbb{Z} )$.
Here, we denote by $\partial_x, \partial_z$ the $x$-boundary and $z$-boundary in their context.
By Theorem \ref{201907200015}, we obtain an isomorphism of short abstract complexes $\Xi^1 ( K , \partial_x K ; \mathbb{Z} )^T \cong  \Xi^1 ( K^\vee , (\partial_z K)^\vee ; \mathbb{Z} )$.
Under the isomorphism, the $(\pm)$-stabilizers correspond to each other based on Corollary \ref{201907221136}.
\end{Example}

\begin{Example}
Consider a $\mathbb{Z} /2$-action $\phi$ on $A$.
Here, we deal with $\mathbb{Z}/2$ as a ring.
Note that any manifold is $\mathbb{Z}/2$-oriented.
It implies that the duality of the $(\pm)$-stabilizers and the elementary operator holds for any manifold.
\end{Example}

\begin{Example}
The $\mathbb{Z}$-orientability is equivalent with the usual orientability of manifolds.
Recall that any bicommutative Hopf algebra $A$ has a canonical $\mathbb{Z}$-action.
We obtain the duality of the $(\pm)$-stabilizers and the elementary operator with respect to oriented manifolds for any finite-dimensional bisemisimple bicommutative Hopf algebra $A$.
\end{Example}

\vspace{0.5cm}


\begin{appendices}


\section{Poincar\'e-Lefschetz duality}
\label{201907211845}

\begin{Defn}
Let $T$ be a topological space and $T^{(\bullet)}$ be a filtration of $T$, i.e. a sequence of subspaces,
\begin{align}\notag
\cdots
\subset
T^{(-1)}
\subset
T^{(0)}
\subset
T^{(1)}
\subset
\cdots
\subset
T^{(q)}
\subset
\cdots 
\subset
T
.
\end{align}
We do not assume $\bigcup_{q} T^{(q)} =T$ or $\bigcap_{q} T^{(q)} = \emptyset$.

We define a {\it chain complex associated with a filtration $T^{(\bullet)}$} denoted by
\begin{align}\notag
C^{filt}_q ( T^{(\bullet)} ) \stackrel{\mathrm{def.}}{=} H_q ( T^{(q)}, T^{(q-1)}; R ) 
\end{align}
, the $q$-th singular homology theory of a pair of topological spaces $(T^{(q)} , T^{(q-1)} )$ with coefficients in $R$.
The boundary homomorphism $\partial^{filt}_q :C^{filt}_q (T^{(\bullet)} ) \to C^{filt}_{q-1} (T^{(\bullet)} )$ is defined by a composition of 
\begin{align}\notag
H_q (T^{(q)}, T^{(q-1)} ; R ) \stackrel{\partial_q}{\to} H_{q-1} (T^{(q-1)} ; R )  \to H_{q-1} ( T^{(q-1)} , T^{(q-2)} ; R ) . 
\end{align}
Here, the latter homomorphism is induced by the inclusion $(T^{(q-1)}, \emptyset) \to (T^{(q-1)} , T^{(q-2)})$.
Note that $\partial^{filt}_{q-1} \circ \partial^{filt}_q = 0$ \cite{massey}.
We define a chain complex $C^{filt}_\bullet (T^{(\bullet)} ) \stackrel{\mathrm{def.}}{=} \left( C^{filt}_q (T^{(\bullet)} ) , \partial^{filt}_q \right)_{q\in \mathbb{Z}}$.
We denote by $H^{filt}_\bullet ( T^{(\bullet)} )$ the homology theory associated with the chain complex $C^{filt}_\bullet (T^{(\bullet)} )$.
\end{Defn}

\begin{prop}
\label{201907171109}
Suppose that $H_p ( T^{(q)}, T^{(q-1)} ; R ) \cong 0$ if $p \neq q$.
Then we have a natural isomorphism between the homology theory associated with the filtration and the singular homology theory, $H^{filt}_q ( T^{(\bullet)} ) \cong \varinjlim_{k} H_q ( T^{(k)} ; R )$.
\end{prop}
\begin{proof}
We sketch the proof.
By assumption, the induced morphism $H_q ( T^{(q+1)} ; R ) \to H_q ( T^{(p)} ; R )$ is an isomorphism for $p > q$.
In particular, we have $\varinjlim_{k} H_q ( T^{(k)} ; R ) \cong H_q (T^{(q+1)} ; R)$.
By the assumption, $H_q (T^{(q-1)} ; R ) \cong 0$.
Hence, the kernel of $\partial^{cell}_q$ is isomorphic to the kernel of $H_q (T^{(q)}, T^{(q-1)} ; R ) \to H_{q-1} ( T^{(q-1)} ; R )$ which coincides with the image of $H_q (T^{(q)} ; R ) \to H_q (T^{(q)}, T^{(q-1)} ; R )$ by a long exact sequence.
The image is isomorphic to $H_q (T^{(q)} ; R )$ since $H_{q-1} ( T^{(q+1)} ; R ) \cong 0$ by the assumption.
As a result, the kernel of $\partial^{cell}_q$ is isomorphic to $H_q (T^{(q)} ; R )$.
Hence, the $q$-th cellular homology is given by a cokernel of $H_{q+1} (T^{(q+1)} , T^{(q)} ; R ) \to H_q (T^{(q)} ; R)$ which is isomorphic to $H_q (T^{(q+1)} ; R )$ since $H_q (T^{(q+2)}; R ) \cong 0$ by the assumption.
\end{proof}

We extend the notion of {\it regular CW-complex}, its cellular chain complex and homology as follows.
They are necessary to describe a Poincar\'e-Lefschetz duality on chain complex level.

\begin{Defn}
A topological space $X$ is a {\it homology $k$-sphere} if there exists an isomorphism $H_\bullet ( X ; \mathbb{Z}) \cong H_\bullet (S^k ; \mathbb{Z})$.
Note that by the long exact sequence, we have $H_\bullet ( CX , X ; \mathbb{Z}) \cong H_\bullet ( D^{k+1} , S^k ; \mathbb{Z})$ where $CX$ is the cone of $X$.
\end{Defn}

\begin{Defn}
\label{201907251333}
A {\it regular homology q-cell} on a topological space $X$ is given by a pair $(c, \varphi)$ such that 
\begin{itemize}
\item
$c \subset X$ is a subspace.
Denote by $\bar{c}$ the closure and $\partial c$ the boundary.
\item
The boundary of $\partial c$ is a homology $(q-1)$-sphere.
\item
$\varphi$ is a homeomorphism of pairs $( \bar{c} , \partial c ) \cong ( C \partial c_{q}, \partial c_{q})$ preserving $\partial c_{q} $.
\end{itemize}
For simplicity, we often omit a regular homology $q$-cell $(c, \varphi)$ to $c$.
\end{Defn}

\begin{Defn}
An {\it orientation} on a regular homology $q$-cell is defined by an isomorphism $H_q ( \bar{c} , \partial c ; \mathbb{Z} ) \cong \mathbb{Z}$.
An {\it oriented regular homology $q$-cell} is a regular homology $q$-cell equipped with an orientation.
\end{Defn}

\begin{Example}
A regular $q$-cell of a regular CW-complex gives a regular homology $q$-cell.
\end{Example}

\begin{Defn}
\label{20190722125}
A {\it regular complex structure $K$ with homology cells} on $X$ is given by a family of homology cells $\{ c_{q,j} \}_{q,j}$ in $X$ such that
\begin{itemize}
\item
$c_{q,j}$ is a regular homology $q$-cell for $j = 1, 2, \cdots , n_q$.
\item
$c_{q,j} \bigcap c_{q^\prime,j^\prime} = \emptyset$
\item
$\coprod_{q,j} \overline{c}_{q,j} \to X$ is a quotient map.
\item 
If we denote by $K^{(q)}$ the $q$-skeleton, i.e. $K^{(q)}$ is a union of homology $r$-cells for $r \leq q$, then $\partial c_{q,j} \subset K^{(q-1)}$.
\end{itemize}

A {\it regular complex with oriented homology cells} is a regular complex with homology cells whose homology cells are oriented.

For a regular complex structure $K$ with homology cells on $X$, $L$ is a {\it subcomplex} if $L$ is closed set in $X$ and $L$ is the union of a set of homology cells in $K$.
Then the homology cells consisting of $L$ gives a regular complex structure with homology cells.
\end{Defn}

\begin{Example}
A regular CW-complex structure gives a regular complex structure with homology cells.
\end{Example}

Fix a closed $R$-oriented $m$-dimensional manifold $M$ which is triangulable.
Recall the Poincar\'e-Lefschetz duality :
We have an isomorphism
\begin{align}\notag
[M] \cap : H^q ( M ; R) \to H_{m-q} ( M ; R ) . 
\end{align}
Note that the isomorphism is induced by the cap product with the fundamental class of $M$, $[M] \in H_m (M ; R)$.

\begin{Defn}
\label{201907251321}
Let $\Delta$ be a polyhedral complex structure of $M$ and $\Delta^\prime$ be the first barycentric subdivision.
Let $K$ be a subcomplex of $\Delta$.
Then the {\it supplement of $K$ in $\Delta$}, denoted by $K^\ast$, is defined as a subcomplex of $\Delta^\prime$, which consists of all simplcies none of whose vertices are in $K^\prime$.

Consider the skeleton filtration of $K$,
\begin{align}\notag
K^{(\bullet)} : 
\emptyset = 
\cdots
=
K^{(-1)}
\subset
K^{(0)}
\subset
K^{(1)}
\subset
\cdots
\subset
K^{(q)}
\subset
\cdots 
\subset
K^{(m)}
= 
\cdots
=
K . 
\end{align}
\label{201907251337}
Define a filtration $T^{(\bullet)} (K)$ of $\Delta^\prime$ by $T^{(q)} (K)\stackrel{\mathrm{def.}}{=} \left( K^{(n-q-1)} \right)^\ast$ where
\begin{align}\notag 
T^{(\bullet)}(K) : 
K^\ast
=
\cdots
=
\left( K^{(m)} \right)^\ast 
\subset
\cdots
\subset
\left( K^{(0)} \right)^\ast 
\subset
\left( K^{(-1)} \right)^\ast 
=
\cdots
=
\Delta^\prime
\end{align}
\end{Defn}

We have a refinement of the Poincar\'e-Lefschetz duality (\cite{whitehead} Theorem 7.4),
\begin{align}
\label{201907260248}
H^q ( L , K ; R) \cong H_{m-q} ( L^\ast , K^\ast ; R ) . 
\end{align}

Since the $q$-skeleton $K^{(q)}$ is also a subcomplex of $\Delta$, we have,
\begin{align}
\label{201907202107}
H^q ( K^{(q)} , K^{(q-1)} ; R ) \cong H_{m-q} ( \left( K^{(q-1)} \right)^\ast , \left( K^{(q)} \right)^\ast ; R ) . 
\end{align}

Recall that the left hand side $H^q ( K^{(q)} , K^{(q-1)} ;R)$ is naturally isomorphic to the $q$-th component of cellular cochain complex $C^q_{cell} ( K )$.
The right hand side is the filtration chain complex of $T^\bullet (K)$ in Definition \ref{201907251321}, in particular we obtain an isomorphism,
\begin{align}
\label{201907212246}
\varrho^q : C^q_{cell} ( K ) \to C^{filt}_{m-q} ( T^{(\bullet)} (K) )  . 
\end{align}

\begin{prop}
\label{201907251507}
The sequence of isomorphisms $\tilde{\varrho} = \left( (-1)^{q(m+1)} \varrho^q \right)_{q \in \mathbb{Z}}$ gives a chain isomorphism 
\begin{align}
\label{201907212310}
\tilde{\varrho} : C^\bullet_{cell} (K ) \to C^{filt}_{m-\bullet} ( T^{(\bullet)} (K) ) .
\end{align}
The isomorphism is natural for inclusions :
Let $L \subset K$ be a subcomplex.
Then the following diagram commutes,
\begin{equation}\notag
\begin{tikzcd}
C_{cell}^\bullet ( K  )
\ar[r]
\ar[d, "\tilde{\varrho} "]
&
C_{cell}^\bullet ( L )
\ar[d, "\tilde{\varrho} "]
\\
C^{filt}_{m - \bullet} (T^{(\bullet)} (K)  )
\ar[r]
&
C^{filt}_{m-\bullet} ( T^{(\bullet)} (L) ) 
\end{tikzcd}
\end{equation}
\end{prop}
\begin{proof}
The first claim follows from Theorem 6.31 \cite{whitehead}.
The second claim is due to Corollary 6.33 \cite{whitehead}.
\end{proof}

We introduce a notion of {\it dual cell} to deal with (\ref{201907212310}) as a cellular chain complex in the extended sense.

\begin{Defn}
Recall $\Delta$ and $\Delta^\prime$ in Definition \ref{201907251321}.
Let $c_q$ be a $q$-cell of $\Delta$.
We define the {\it dual cell of $c_q$} as a subspace $c^\vee_q \subset M$ to be a union of simplices consisting of joins $t_q \hat{c}_q + t_{q+1} \hat{c}_{q+1} + \cdots + t_m \hat{c}_m$ where $t_q \neq 0$ and $c_j$'s runs all $j$-cells such that $c_j$ is a face of $c_{j+1}$.
Here, $m$ is the dimension of $M$ as before.
\end{Defn}

\begin{Defn}
\label{201907260316}
We define an isomorphism $\kappa^q : H^q ( c_q , \partial c_q ; R ) \cong H_{m-q} ( c^\vee_q , \partial c^\vee_q ; R )$ by
\begin{align}\notag
H^q ( c_q , \partial c_q ; R )
\stackrel{(-1)^{q(m+1)}}{\to}
H^q ( c_q , \partial c_q ; R )
\stackrel{\mathrm{PD}}{\to}
H_{n-q} ( (\partial c_q)^\ast , (c_q)^\ast ; R )
\stackrel{\mathrm{exc.}}{\to}
 H_{m-q} ( c^\vee_q , \partial c^\vee_q ; R )
\end{align}
Here, PD means (\ref{201907260248}).
In particular, the boundary $\partial c^\vee_q$ is a homology $(n-q-1)$-sphere.
\end{Defn}

\begin{Lemma}
\label{201907212211}
A dual cell $c^\vee_q$ of a $q$-cell $c_q$ is naturally a regular homology $(n-q)$-cell on $M$.
Furthermore, the orientation on a cell $c_q$ by its characteristic map induces an orientation on the dual cell $c^\vee_q$ via the isomorphism $\kappa^q$.
\end{Lemma}
\begin{proof}
We compute the closure and boundary of $c^\vee_q$ in $M$.
Note that taking closures preserves finite unions.
The closure of $c^\vee_q$ in $M$ is given by a union of simplices whose vertices are $\hat{c}_{q}, \hat{c}_{q+1}, \cdots , \hat{c}_m$ for all $j$-cells $c_{j}$ such that $c_j$ is a face of $c_{j+1}$ for $q \leq j \leq m$.
Hence, the boundary of $c^\vee_q$ in $M$, $\partial c^\vee_q$, is a union of  simplices whose vertices are $\hat{c}_{q+1}, \cdots , \hat{c}_m$ for all $j$-cells $c_{j}$ such that $c_j$ is a face of $c_{j+1}$ for $q \leq j \leq m$.
In particular, the closure $\overline{c^\vee_q}$ is a cone of the boundary $\partial c^\vee_q$.
\end{proof}

\begin{remark}
The dual cell $c^\vee_q$ is not a cell, i.e. $(\bar{c}^\vee_q , \partial c^\vee_q)$ is not homeomorphic to $(D^{n-q}, S^{n-q-1})$ in general.
Even if $q=0$, for example, then $\partial c^\vee_0$ is a link of the vertex $c_0$ in $\Delta^\prime$.
Note that the link is not homeomorphic to a sphere since there exists a topological manifold with a triangulatioin which does not admit a PL structure for $m \geq 5$.
\end{remark}

\begin{prop}
The dual-cells for $\Delta$ give a regular cell-complex structure with oriented homology cells on $M$.
We denote the regular complex by $\emptyset^\star$.
\end{prop}
\begin{proof}
By Lemma \ref{201907212211}, the dual-cells are homology cells.
The dual cells are disjoint to each other by definitions.
Note that the dual cells are subcomplexes of the finite simplicial complex $\Delta^\prime$.
Hence, the gluing map is a quotient map.
If we denote by $(\emptyset^\star)^{(q)}$ the $q$-skeleton, then the boundary $\partial c^\vee_q$ of dual-cell lies in $(\emptyset^\star)^{(n-q-1)}$ by the computation of the boundary $\partial c^\vee_q$ in the proof of Lemma \ref{201907212211}.
It completes the proof.
\end{proof}

\begin{Defn}
\label{202007062100}
Let $K$ be a subcomplex of $\Delta$ as before.
Then the collection of dual-cells lying in $K^\ast$ is a subcomplex of $\emptyset^\ast$ in the sense of Definition \ref{20190722125}.
We denote the induced regular cell-complex structure by $K^\star$.
\end{Defn}

Recall the filtration $T^{(\bullet)} (K)$ in (\ref{201907251337}).Then by definitions, we have an equality between filtrations,
\begin{align}\notag
T^{(\bullet)} (K) = \left( \emptyset^\star , K^\star \right)^{(\bullet)} .  
\end{align}
Here, $\left( \emptyset^\star , K^\star \right)^{(\bullet)}$ is the skeleton filtration associated with the complex-pair $\left( \emptyset^\star , K^\star \right)$, i.e. $\left( \emptyset^\star , K^\star \right)^{(q)} \stackrel{\mathrm{def.}}{=} K^\star \bigcup (\emptyset^\star)^{(q)}$.
The chain isomorphism $\tilde{\varrho}$ induces a chain isomorphism $C^\bullet_{cell} ( K ) \cong C^{filt}_\bullet ( (\emptyset^\star, K^\star)^{(\bullet)} )$

\begin{prop}
\label{201907251509}
The chain isomorphism $\tilde{\varrho}$ induces a chain isomorphism,
\begin{align}
\label{201907251443}
C^\bullet_{cell}(K ) \cong C^{cell}_{m-\bullet} ( \emptyset^\star , K^\star  ) . 
\end{align}
Note that the right hand side is defined by a cokernel of $C^{cell}_{m - \bullet} ( K^\star  ) \to C^{cell}_{m - \bullet} ( \emptyset^\star )$.

The isomorphism is natural for inclusions :
Let $L \subset K$ be a subcomplex.
Then the following diagram commutes,
\begin{equation}\notag
\begin{tikzcd}
C_{cell}^\bullet ( K  )
\ar[r]
\ar[d, "\tilde{\varrho} "]
&
C_{cell}^\bullet ( L  )
\ar[d, "\tilde{\varrho} "]
\\
C^{cell}_{m - \bullet} ( \emptyset^\star , K^\star  )
\ar[r]
&
C^{cell}_{m-\bullet} ( \emptyset^\star , L^\star  ) 
\end{tikzcd}
\end{equation}
\end{prop}
\begin{proof}
It suffices to prove that $C^{cell}_\bullet (  \emptyset^\star  , K^\star ) \cong C^{filt}_\bullet ( (\emptyset^\star, K^\star)^{(\bullet)}  )$.
Recall that $C^{cell} ( K^\star   ) = C^{filt} ( (K^\star)^{(\bullet)}  )$ by definition.
It follows from the induced short exact sequence.
\begin{align}\notag
0
\to
C^{filt}_\bullet ( ( K^\star)^{(\bullet)}  )
\to
C^{filt}_\bullet ( (\emptyset^\star )^{(\bullet)} )
\to
C^{filt}_\bullet ( (\emptyset^\star, K^\star)^{(\bullet)} )
\to
0
\end{align}

The commutativity in the second claim follows from that of Proposition \ref{201907251507}.
\end{proof}

\begin{Corollary}
\label{201907211913}
Let $L \subset K$ be subcomplexes of $\Delta$.
Then we have a chain isomorphism which extends (\ref{201907251443}). 
\begin{align}
\label{201907260320}
C_{cell}^\bullet ( L , K  ) \cong C^{cell}_{n-\bullet} ( L^\star , K^\star  ) .
\end{align}
Under the orientations of dual cells in Lemma \ref{201907212211}, the isomorphism assigns $c^\vee_q$ to $c_q$.
\end{Corollary}
\begin{proof}
The isomorphism (\ref{201907260320}) follows from the naturality with respect to inclusions in Proposition \ref{201907251509}.
The second claim follows from the definition of $\kappa_q$ in Definition \ref{201907260316}.
\end{proof}
\end{appendices}

\bibliography{Hamiltonian_model_induced_by_chain_complex_theory_3}{}
\bibliographystyle{plain}

\end{document}